\documentclass[prd,amsmath,superscriptaddress,twocolumn,nofootinbib]{revtex4-2}
\pdfoutput=1

\usepackage{textcomp}
\usepackage[utf8]{inputenc}
\usepackage{graphicx}
\usepackage{float}
\usepackage{epsf,latexsym,bbm,euscript}
\usepackage{amssymb,amsmath}
\usepackage{mathtools} 
\usepackage{times}
\usepackage{soul,xcolor}
\usepackage{mathtools}
\usepackage{mathrsfs}
\usepackage{array}
\usepackage{booktabs}
\usepackage{enumitem}
\usepackage{tensor}
\usepackage[dvipsnames]{xcolor}
\usepackage{tikz}
\usetikzlibrary{calc}

\usepackage[caption=false]{subfig}

\usepackage{amsthm}
\newtheorem{theorem}{Theorem}
\newtheorem{corollary}{Corollary}

\newcommand{\RN}[1]{\uppercase\expandafter{\romannumeral #1\relax}}

\usepackage{environ}
\NewEnviron{thincases}{\scalebox{0.5}[1]{$\displaystyle
		\left\{\scalebox{2}[1]{\setlength{\arraycolsep}{6pt}
			$\displaystyle\begin{array}{ll}
				\BODY
			\end{array}$}\right.$}}%}
			\NewEnviron{thinpmatrix}{\scalebox{0.5}[1]{$\displaystyle
	\left(\scalebox{2}[1]{$\displaystyle
		\begin{matrix}
			\BODY
		\end{matrix}$}
	\right)$}}
	
	\makeatletter
	\g@addto@macro\bfseries{\boldmath}
	\makeatother
	
	\makeatletter
	\newcommand*{\defeq}{\mathrel{\rlap{%
		\raisebox{0.3ex}{$\m@th\cdot$}}%
	\raisebox{-0.3ex}{$\m@th\cdot$}}%
=}
\newcommand*{\eqdef}{=\mathrel{\rlap{%
		\raisebox{0.3ex}{$\m@th\cdot$}}%
	\raisebox{-0.3ex}{$\m@th\cdot$}}%
	}
	\makeatother

	\usepackage{pifont}

	\usepackage{scalerel}
	\usepackage{tikz}
	\usetikzlibrary{svg.path}
	\definecolor{orcidlogocol}{HTML}{A6CE39}
	\tikzset{
orcidlogo/.pic={
	\fill[orcidlogocol] svg{M256,128c0,70.7-57.3,128-128,128C57.3,256,0,198.7,0,128C0,57.3,57.3,0,128,0C198.7,0,256,57.3,256,128z};
	\fill[white] svg{M86.3,186.2H70.9V79.1h15.4v48.4V186.2z}
	svg{M108.9,79.1h41.6c39.6,0,57,28.3,57,53.6c0,27.5-21.5,53.6-56.8,53.6h-41.8V79.1z M124.3,172.4h24.5c34.9,0,42.9-26.5,42.9-39.7c0-21.5-13.7-39.7-43.7-39.7h-23.7V172.4z}
	svg{M88.7,56.8c0,5.5-4.5,10.1-10.1,10.1c-5.6,0-10.1-4.6-10.1-10.1c0-5.6,4.5-10.1,10.1-10.1C84.2,46.7,88.7,51.3,88.7,56.8z};
}
}
\newcommand\orcidlink[1]{\href{https://orcid.org/#1}{\mbox{\scalerel*{
			\begin{tikzpicture}[yscale=-1,transform shape]
				\pic{orcidlogo};
			\end{tikzpicture}
		}{X}}}}
		
		\usepackage{url,hyperref}
		\hypersetup{colorlinks,linkcolor={blue!55!black},citecolor={red!50!black},urlcolor={blue!45!black},breaklinks=true}
		
		\usetikzlibrary{arrows.meta,calc,decorations.markings}

		\usepackage{url,hyperref}
		\hypersetup{colorlinks,linkcolor={blue!55!black},citecolor={red!50!black},urlcolor={blue!45!black}}
		
\begin{document}

\title{Thermodynamics of polymerized vacuum regular black holes\\ in anti-de Sitter spacetime}

\author{Sepideh Bakhoda\orcidlink{0000-0002-1926-7712}}
\email{s.bakhoda@gmail.com}
\affiliation{Institute for Theoretical Physics and Cosmology, Zhejiang University of Technology, Hangzhou, 310023, China}
\affiliation{Institute for Theoretical Sciences and Departments of Physics and Astronomy, Westlake University, Hangzhou 310030, China}

\author{Ioannis Soranidis\orcidlink{0000-0002-8652-9874}}
\email{ioannis.soranidis@westlake.edu.cn}
\affiliation{Institute for Theoretical Sciences and Departments of Physics and Astronomy, Westlake University, Hangzhou 310030, China}
\affiliation{Institute of Natural Sciences, Westlake Institute for Advanced Study, Hangzhou 310024, China}

\begin{abstract}
	\vspace*{0.2cm}	
	
We derive a class of vacuum regular black holes inspired by effective loop quantum gravity dynamics and extend the construction to asymptotically anti-de Sitter spacetimes. The derivation is based on a deparameterized Lemaître--Tolman--Bondi formulation, where an auxiliary dust field is introduced only to define an internal time and does not act as a matter source. In spherical symmetry, the dynamics reduces to a set of independent radial shells, giving rise to a factorized shell Hamiltonian and to a Birkhoff-type property: for a fixed reconstruction function and cosmological constant, the static geometry is uniquely determined by the mass. Within this framework, we establish the conditions for curvature regularity at the center and construct several regular black hole models with de Sitter cores and corresponding models with anti-de Sitter cores. We then study their thermodynamics in the extended phase space, using the Hawking--Page transition to compare the de Sitter- and anti-de Sitter-core branches and show that their quantitative differences arise from the deformation of the physical outer-horizon branch rather than from regularity alone.

\end{abstract}
\maketitle

\section{introduction}

Classical gravitational collapse exposes one of the sharpest limitations of general relativity (GR): under broad conditions the theory predicts black hole interiors ending at curvature singularities, where the classical spacetime description loses predictive power \cite{H:76}. This motivates the search for compact object geometries that retain the successful exterior features of black holes while avoiding a singular center. Broadly, such proposals either modify the interior geometry, for example by replacing the Schwarzschild singularity with a regular core, or modify the global structure so that the object is horizonless \cite{CP:19,BCNS:19,M:23}.

Regular black hole (RBH) solutions can be constructed through a variety of mechanisms, either by introducing suitable matter sources or by modifying the gravitational dynamics. A well-known example of the former arises in nonlinear electrodynamics, where magnetic monopole configurations can regularize the spacetime within GR \cite{pt:69,bmss:79,ag:98,ag:99a,ag:99b,b:00,b:01,bh:02,d:04,bv:14,fw:16,b:17,tsa:18,b:23}. More recently, considerable attention has been devoted to vacuum regular black holes, where regularity is achieved without introducing additional matter sources. Such solutions have been obtained, for instance, in quasitopological gravity: in dimensions $D\geqslant 5$ for polynomial theories \cite{BCH:25,BCHM:25a,BCHM:25b,BCHMC:25,HKMS:25,FKSZ:25}, and in four dimensions through non-polynomial constructions \cite{BCR:26}\footnote{More generally, the relation between integrable two-dimensional Horndeski/dilaton theories and $D\geqslant 4$ gravitational theories admitting Birkhoff-type static vacuum solutions has recently been clarified in Ref.~\cite{B:26}. A general analysis of the thermodynamics of quasitopological gravity with flat or AdS asymptotics can also be found in Ref.~\cite{B:26-thermo}.}. Analogous four-dimensional vacuum configurations have also been derived in related frameworks \cite{BenAchour:2018khr,Han:2020uhb,Han:2022rsx,F:25,BCHM:26,GLRSW:25,GLSW:25,GL:25,LS:26}.

A complementary approach, more directly motivated by quantum gravity, is provided by loop quantum gravity-inspired effective models. In this setting, quantum corrections are introduced through modifications of the scalar constraint, leading to an effective canonical dynamics. Spherical symmetry is especially useful for implementing this idea: since gravitational waves are absent, the dynamics can be organized in terms of radial shells. This makes it possible to construct vacuum models in which the shells evolve independently once the relevant conserved quantities are fixed, and whose static sector can admit unique black hole solutions \cite{GL:25}.

The construction considered in this work is based on a deparameterized formulation, where an auxiliary dust field is introduced only to define a relational time. The dust does not act as physical matter sourcing the geometry. Instead, the vacuum sector is obtained by setting its energy density to zero, so that the resulting configurations should be interpreted as polymerized vacuum solutions. Effective covariant descriptions of related modified dynamics can also be formulated within mimetic gravity-type theories \cite{Chamseddine:2016uef,Langlois:2017hdf,BenAchour:2017ivq,Han:2020uhb,Han:2022rsx,GLSW:25,GL:25}, providing a useful bridge between the canonical and spacetime-based descriptions. Explicit realizations of these constructions, together with preliminary analyses of their stability, can be found in Refs.~\cite{GL:25,LS:26,LS:26-mass-inflation}.

A key feature of the vacuum sector obtained in this way is its Birkhoff-type behaviour. Once the reconstruction function is fixed, the static spherically symmetric geometry is uniquely determined by the mass parameter. This is a nontrivial restriction on the allowed polymerizations. For a generic choice, the reconstructed metric may depend on additional shell data, in which case the vacuum geometry would not be specified by the mass alone. The Birkhoff-type condition removes this ambiguity by requiring the curvature variable entering the static reconstruction to depend only on the conserved vacuum shell Hamiltonian \cite{GL:25,LS:26}. As a result, the metric function is determined by a single polymerization function, together with the mass. In the classical limit, the standard Schwarzschild geometry is recovered.

In this work, we extend the polymerized vacuum construction to include a cosmological constant. This extension is important for two complementary reasons. First, it allows the resulting regular geometries to be embedded in spacetimes with nonzero cosmological constant, thereby making the construction more general. While the observed large-scale acceleration of the universe is usually described by an effective positive cosmological constant, negative cosmological constant backgrounds play a central role in high-energy theory, most notably in the context of gauge--gravity duality \cite{M:99}. It is therefore natural to ask whether the polymerized vacuum framework admits a consistent extension to asymptotically anti-de Sitter (AdS) spacetimes.
	
	Second, the AdS case provides a particularly useful setting for black hole thermodynamics. In contrast with asymptotically flat spacetimes, the AdS boundary effectively confines radiation, allowing black holes to be studied in a canonical ensemble and making possible the Hawking--Page transition between thermal AdS and a large black hole. This provides a controlled framework in which to investigate how polymerization and singularity resolution affect the thermodynamic phase structure \cite{KMT:17}.
	
	Another important feature of the present construction is that it naturally accommodates regular black holes with both de Sitter-like and anti-de Sitter-like cores. Most regular models replace the central singularity with a dS-like core, as in the Bardeen, Hayward, and Dymnikova geometries. Regularity, however, does not by itself require a dS interior. A nonsingular geometry may instead possess an AdS-like core, which, from an effective stress-tensor perspective, can be interpreted as being supported by a negative energy density near the origin. Such behavior can arise from quantum backreaction effects in highly compact objects \cite{ALNV:25} and has also appeared in other settings, although typically through the inclusion of additional matter fields, as in certain modified-gravity and nonlinear-electrodynamics constructions \cite{HKMS:25}.
	
	In the polymerized vacuum setting considered here, both possibilities can be realized through suitable choices of the polymerization function. This makes the framework well suited for comparing how different vacuum resolutions of the singularity affect black hole thermodynamics. To the best of our knowledge, such a comparison has not previously been carried out for polymerized vacuum regular black holes in AdS spacetime.

The paper is organized as follows. In Sec.~\ref{sec:LTB}, we review the Lemaître–Tolman–Bondi (LTB) shell formulation, introduce the cosmological constant, and derive the polymerized vacuum reconstruction. We also discuss how regularity at the center is encoded in the properties of the polymerization function, and present the representative models used throughout the paper. In Sec.~\ref{sec:thermo}, we derive the temperature, entropy, and extended first law for polymerized AdS black holes. In Sec.~\ref{sec:PT}, we study the phase structure and equation of state. We conclude in Sec.~\ref{sec:conclusions} with a discussion of the main results and possible extensions.

\section{Polymerized vacuum solutions with a cosmological constant} \label{sec:LTB}

We first summarize the ingredients of the LTB framework that are needed for the present analysis. The purpose of this review is not to reproduce the full derivation, but to isolate the structural properties that will later be used to build effective quantum gravity-inspired modifications, construct RBH geometries, and study their thermodynamic properties. Detailed treatments of the LTB reduction and its effective extensions can be found in Refs.~\cite{LLB:07,GLRSW:25,GLSW:25,LS:26}.

\subsection{Classical LTB model}\label{sec:CLTB}

We consider a spherically symmetric gravitational system coupled to pressureless dust. Using Ashtekar--Barbero variables adapted to spherical symmetry, the gravitational phase space is described by the canonical pairs $(E^{x},K_{x})$ and $(E^{\phi},K_{\phi})$. Here $E^{I}$ denote the densitized triad variables, while $K_{I}$ encode the corresponding components of the extrinsic curvature\footnote{On a hypersurface $t=\mathrm{const}$ with induced metric $h_{ab}$, one may introduce triads $e^{i}_{a}$ satisfying $h_{ab}=e^{i}_{a}e^{j}_{b}\delta_{ij}$. This description carries a local $\mathrm{SO}(3)$ gauge freedom, inherited from the local Lorentz symmetry of the tetrad formulation. The extrinsic curvature can similarly be written in triadic form through $k^{a}_{i}e^{i}_{b}=k_{ab}$. The inverse densitized triads are defined by $E^{a}_{i}=\frac{1}{2}\epsilon_{ijk}\epsilon^{abc}e^{j}_{b}e^{k}_{c}$ \cite{Rov:10,RV:15,T:10}.}. The most general spherically symmetric line element in these variables is
\begin{align}
	ds^2=-N^2dt^2+\frac{(E^{\phi})^2}{E^x}(dx+N^xdt)^2+E^x d\Omega^2 ,
\end{align}
where $N(t,x)$ is the lapse and $N^x(t,x)$ is the radial shift. The dust sector is taken to be pressureless, so its stress tensor has the standard form
\begin{align}
	T_{\mu\nu}=\varrho u_\mu u_\nu ,
\end{align}
with $\varrho$ the dust energy density and $u^\mu$ the dust four-velocity. We now make the coordinate choices appropriate to the LTB description. The radial coordinate is chosen to label the dust worldlines, so that the dust flow is aligned with curves of constant $x$. In these coordinates the four-velocity may be written as
\begin{align}
	u^\mu=\delta^\mu_t .
\end{align}
We then work in the comoving gauge,
\begin{align}
	N=1,
	\quad
	N^x=0,
\end{align}
for which the line element becomes
\begin{align}
	ds^2=-dt^2+\frac{(E^{\phi})^2}{E^x}dx^2+E^x d\Omega^2 .
	\label{eq:metric}
\end{align}
In this gauge, $t$ measures the proper time along the radial dust geodesics, while $x$ labels the individual shells. Thus the coordinates are adapted to a freely falling reference congruence. This description is assumed to hold only in regions where the congruence remains regular, so that the dust time and shell label define a valid coordinate system. If shell crossings or caustics occur, they should be understood as a failure of this reference-clock chart. They do not represent an instability of physical dust matter in the vacuum sector studied below, since there the dust density is set to zero and the dust field is retained only as a relational clock.

The intrinsic curvature of the spatial slices is characterized, in spherical symmetry, by two independent sectional curvatures. We define
\begin{align}
	\mathcal{R}_1(t,x)\equiv-\tensor{R}{^{\theta\phi}_{\theta\phi}}, \quad \mathcal{R}_2(t,x)\equiv\tensor{R}{^{x\theta}_{x\theta}}.
\end{align}

With this convention, $\mathcal{R}_1$ is associated with the angular two-spheres, while $\mathcal{R}_2$ corresponds to the radial-angular sectional curvature. These quantities provide a convenient geometric parametrization of the intrinsic spatial curvature and will be used both in the classical LTB reduction and in the effective modifications considered below.
	The LTB line element can be written as
	\begin{align}
		ds^2=-dt^2+\frac{\big(R'(t,x)\big)^2}{1+\mathcal{E}(x)}dx^2+R(t,x)^2d\Omega^2 .
		\label{eq:LTB:metric}
	\end{align}
	Comparing this metric with Eq.~\eqref{eq:metric} fixes the relation between the triad variables $E^x$ and $E^\phi$. Since $E^x=R^2$, consistency of the radial components gives the LTB condition
	\begin{align}
		(E^x)'=2\sqrt{1+\mathcal{E}(x)}\,E^\phi .
	\end{align}
	
	The function $\mathcal{E}(x)$ arises as an integration function of the LTB system and is conserved along the dust flow. It also has a direct geometrical meaning. Using the definitions of the intrinsic curvature variables introduced above, one obtains
	\begin{align}
		\mathcal{R}_1(t,x)=\frac{\mathcal{E}(x)}{E^x},
		\quad
		\mathcal{R}_2(t,x)=-\frac{\mathcal{E}'(x)}{(E^x)'} .
	\end{align}
	Thus $\mathcal{E}(x)$ controls the intrinsic spatial curvature of the LTB slices.
	
	Physically, $\mathcal{E}(x)$ is interpreted as the specific energy of the shell labeled by \(x\). Its sign distinguishes the usual three classes of LTB evolution:
	\vspace*{0.2cm}
	\begin{itemize}[noitemsep,nolistsep,leftmargin=12pt]
		\item Marginally bound shells: \(\mathcal{E}(x)=0\).  
		These shells have zero total energy and reach infinity with vanishing velocity.
		\vspace*{0.1cm}
		
		\item Unbound shells: \(\mathcal{E}(x)>0\).  
		These correspond to positive-energy shells.
		\vspace*{0.1cm}
		
		\item Bound shells: \(-1<\mathcal{E}(x)<0\).  
		These have negative total energy and do not escape to infinity.
		\vspace*{0.2cm}
	\end{itemize}
	LTB spacetimes are determined by two independent functions of the radial shell label $x$. The first is the energy function $\mathcal{E}(x)$, introduced above, while the second is the Misner--Sharp mass $M(x)$, which measures the total energy contained within the shell $x$. A crucial feature of the pressureless dust system in comoving coordinates is that both of these functions are conserved in time. In particular,
	\begin{align}
		\partial_t\mathcal{E}=0, \quad\partial_t M=0 .
	\end{align}
	These relations are not additional assumptions, but follow from Einstein's equations for pressureless dust. The function $\mathcal{E}(x)$ therefore labels the conserved energy class of each shell, while $M(x)$ gives the conserved mass contained inside it. Since the coordinates are comoving with the dust flow, there is no exchange of energy between neighboring shells. Each shell consequently evolves as an independent subsystem with fixed values of $\mathcal{E}(x)$ and $M(x)$. The independent evolution of the comoving shells is encoded in the radial equation
	\begin{align}
		\dot{R}^2=\frac{2M(x)}{R(t,x)}+\mathcal{E}(x).
	\end{align}
	For a fixed shell label $x$, the quantities $M(x)$ and $\mathcal{E}(x)$ act as constants of motion, so the equation describes the dynamics of that shell without reference to neighboring shells. This is the essential decoupling property of LTB dynamics. Before generalizing the construction beyond GR, it is useful to recast this shell-wise structure in Hamiltonian language. To this end, for each shell we introduce
	\begin{align}
		b(x)=\frac{K_{\phi}}{\sqrt{E^x}},
		\quad
		v(x)=(E^x)^{3/2}=R^3 ,
	\end{align}
	with Poisson bracket
	\begin{align}
		\{b(x),v(y)\}=-\frac{3}{2}\delta(x-y).
	\end{align}
	After imposing the LTB condition, the dynamics can be organized in terms of independent shell variables. The evolution generated by the shell Hamiltonian $H_s(b,v;\mathcal E)$ is therefore
	\begin{align}
		\dot v=\{v,H_s\}=\frac{3}{2}\partial_b H_s,
		\quad
		\dot b=	\{b,H_s\}	=-\frac{3}{2}\partial_v H_s .
		\label{eq:Poisson}
	\end{align}
	The compatibility of these variables with the original triad and extrinsic curvature variables follows once the LTB condition, together with its radial derivative, is imposed. Using these relations, the shell Hamiltonian takes the factorized form
	\begin{align}
		H_s(b,v;\mathcal E)=v H_0(b,\mathcal R_1),\label{eq:H-factor}
	\end{align}
	where \(v=R^3\) in four spacetime dimensions and $\mathcal R_1$ is the intrinsic-curvature variable determined by the LTB energy function \cite{GLRSW:25,GLSW:25,GL:25,LS:26}.
	
	The total scalar constraint is
	\begin{align}
		\mathcal C^{\rm tot}=\mathcal C^{\rm grav}+\mathcal C^{\rm dust}\approx 0 .
	\end{align}
	In the LTB sector, the gravitational part reduces to a radial derivative of the shell Hamiltonian. With the normalization used here, this can be written as
	\begin{align}
		\mathcal C^{\rm grav}=-\frac{2}{\sqrt{1+\mathcal E(x)}}\partial_x H_s(b,v;\mathcal E),\label{eq:Cgrav-shell}
	\end{align}
	up to the common angular factor. The precise overall normalization is conventional and will not affect the vacuum condition below.
	
	On the vacuum branch considered in this work, the dust field is retained only as a relational clock and does not source the geometry. Equivalently, the dust density is set to zero, so that
	\begin{align}
		\mathcal C^{\rm dust}=0 .
	\end{align}
	The total scalar constraint then implies
	\begin{align}
		-\frac{2}{\sqrt{1+\mathcal E(x)}}\partial_x H_s(b,v;\mathcal E)\approx 0 .
	\end{align}
	Therefore, we have that 
	\begin{align}
		\partial_x H_s(b,v;\mathcal E)\approx0 .
	\end{align}
	Thus, in the vacuum sector, the shell Hamiltonian is independent of the shell label \(x\). We may therefore write
	\begin{align}
		H_s(b,v;\mathcal E)=2m,
	\end{align}
	where $m$ is the mass integration constant. Using the factorization in Eq.~\eqref{eq:H-factor}, this is equivalently
	\begin{align}
		H_0(b,\mathcal R_1)=\frac{2m}{R^{3}}.
	\end{align}
	In the GR case without a cosmological constant, one finds explicitly
	\begin{align}
		\mathcal{C}^{\mathrm{grav}}=-\frac{2}{\sqrt{1+\mathcal{E}}}\partial_x\left[R\left(\dot R^2-\mathcal{E}\right)\right] .
	\end{align}
	Therefore the corresponding shell Hamiltonian may be identified, up to the same overall normalization convention, as
	\begin{align}
		H^{\Lambda=0}_{s}=R\left(\dot R^2-\mathcal{E}\right).
	\end{align}
	Let us now include a cosmological constant $\Lambda$. In four spacetime dimensions, the cosmological-constant contribution to the scalar constraint is proportional to the spatial volume density. With our conventions this contribution takes the form
	\begin{align}
		\mathcal{C}^{\Lambda}=2\Lambda E^{\phi}\sqrt{E^x}.
	\end{align}
	
	The LTB metric is given by
	\begin{align}
		ds^2=-dt^2+\frac{(R'(t,x))^2}{1+\mathcal{E}(x)}dx^2+R(t,x)^2d\Omega^2. \label{eq:LTB}
	\end{align}
Demanding that the metric above coincide with that of Eq.~\eqref{eq:metric} necessarily imposes a condition on the functions $E^{\phi}$ and $E^{x}$. We designate this as the LTB condition, and it is readily verified to be
	\begin{align}
		(E^x)'=2E^\phi\sqrt{1+\mathcal{E}},
	\end{align}
	and thus we obtain that 
	\begin{align}
		E^\phi \sqrt{E^x}=\frac{(E^x)'\sqrt{E^x}}{2\sqrt{1+\mathcal{E}}}=\frac{1}{\sqrt{1+\mathcal{E}}}
		\partial_x\left[\frac{(E^x)^{3/2}}{3}\right].
	\end{align}
	Hence the cosmological-constant contribution can be written as
	\begin{align}
		\mathcal{C}^{\Lambda}=\frac{1}{\sqrt{1+\mathcal{E}}}\partial_x	\left(\frac{2\Lambda (E^x)^{3/2}}{3}
		\right)=\frac{1}{\sqrt{1+\mathcal{E}}}\partial_x
		\left(\frac{2\Lambda R^3}{3}\right).
	\end{align}
	Here $\Lambda$ denotes the bare cosmological constant entering the reduced scalar constraint. Since the thermodynamic analysis below concerns black holes in AdS spacetime, we now specialize to\footnote{The restriction to negative $\Lambda$ is made only for the purposes of the AdS thermodynamic analysis. The construction itself does not rely on this sign choice and applies equally to positive bare cosmological constant, leading to the corresponding dS asymptotics whenever the chosen branch admits them.}
	\begin{align}
		\Lambda=-\frac{3}{L^2},
	\end{align}
	where the parameter $L$ is the bare AdS length scale of the classical theory. 
	
	Combining the GR contribution with the cosmological-constant contribution,
	the vacuum scalar constraint becomes
	\begin{align}
		\mathcal{C}^{\mathrm{grav}}+\mathcal{C}^{\Lambda}\approx 0 .
	\end{align}
	Using the expressions above, this may be written as
	\begin{align}
		-\frac{2}{\sqrt{1+\mathcal{E}}}\partial_x\left[R\left(\dot R^2-\mathcal{E}\right)+\frac{R^3}{L^2}	\right]\approx 0 ,
	\end{align}
	up to the same overall normalization convention used for the scalar
	constraint. Therefore the conserved shell Hamiltonian in the AdS case is
	\begin{align}
		H_s^{\mathrm{AdS}}=R\left(\dot R^2-\mathcal{E}\right)+\frac{R^3}{L^2}.
	\end{align}
	The vacuum condition implies
	\begin{align}
		\partial_x H_s^{\mathrm{AdS}}\approx0 .
	\end{align}
	Thus the shell Hamiltonian is constant along the radial direction. Denoting
	this integration constant by $2m$, we have
	\begin{align}
		H_s^{\mathrm{AdS}}=2m .
	\end{align}
	Explicitly,
	\begin{align}
		R\left(\dot R^2-\mathcal{E}\right)+\frac{R^3}{L^2}=2m .
	\end{align}
	Solving for \(\dot R^2\), one obtains
	\begin{align}
		\dot R^2=\mathcal{E}-\frac{R^2}{L^2}+\frac{2m}{R}.
	\end{align}
	The term $-R^2/L^2$ dominates at large radius and confines the shell motion for finite $\mathcal E$ and fixed $m$. This is the shell-level sense in which AdS acts as a gravitational box, and it is precisely the asymptotic structure relevant for black hole thermodynamics.
	
	Equivalently, if one writes the factorized shell Hamiltonian as
	\begin{align}
		H_s^{\mathrm{AdS}}=R^3H_0^{\mathrm{AdS}},
	\end{align}
	then
	\begin{align}
		H_0^{\mathrm{AdS}}=H_0^{(0)}+\frac{1}{L^2}.
	\end{align}
	The vacuum condition $H_s^{\mathrm{AdS}}=2m$ gives
	\begin{align}
		H_0^{\mathrm{AdS}}=\frac{2m}{R^3},
	\end{align}
	and therefore
	\begin{align}
		H_0^{(0)}=-\frac{1}{L^2}+\frac{2m}{R^3}.
	\end{align}
	As we will show below, this is the combination that naturally appears in the polymerized construction.
	
	\subsection{Polymerized case}\label{sec:PLTB}
	
	An appropriate generalization can be carried out in the case of modified theories, under the minimal assumption that compatibility with LTB-type solutions is preserved. By this, we mean that there exists a subset of solutions of LTB type for which a decoupling occurs, allowing each shell, labeled by $x$, to be described independently. The overall procedure remains essentially the same; however, the modified dynamics are now governed by a new gravitational part of the scalar constraint, which must reduce to the GR limit in the appropriate parameter regime \cite{GLSW:25,GLRSW:25,GL:25,LS:26}.
	
	In the dust model, a vacuum region is characterized by the vanishing of the dust energy density. More precisely, by vacuum we mean that the dust density that couples to the deparameterized system is set to zero, so that the dust does not source the metric. From the Hamiltonian perspective, this corresponds to a region in which the shell Hamiltonian, and hence the Misner--Sharp mass $M(x)$, is independent of the shell label $x$, so that $M(x)=\mathrm{const.}$. The energy function $\mathcal{E}(x)$ is not fixed by the vacuum condition itself; it is conserved along each shell but may vary from shell to shell, encoding the choice of LTB reference congruence/slicing.
	
	We should stress that, in the context of a modified theory characterized by a modified scalar constraint, the term ``vacuum solution'' refers to a modified vacuum configuration. In our framework, this is commonly referred to as a polymerized vacuum, reflecting its relation to loop quantum gravity. The reference dust field is nevertheless retained in the covariant action only as the relational clock and foliation variable needed to write the deparameterized Hamiltonian in four-dimensional form; on the vacuum branch it carries no independent physical degree of freedom and should not be interpreted as an additional matter source.
	
	An important issue to clarify is how this construction yields unique static solutions. The procedure proceeds as follows: we start from the LTB metric in Eq.~\eqref{eq:LTB} and perform a coordinate transformation to Painlevé–Gullstrand coordinates,
	\begin{align}
	ds^2=-dt^2+(1+\mathcal{E})^{-1}(dr-\dot R dt)^2+r^2d\Omega^2,
	\end{align}
	with
	\begin{align}
		dr=\dot{R}dt+R'dx.
	\end{align}
	This leads to
	\begin{align}
		ds^2=-dt^2+\frac{1}{1+\mathcal{E}}(dr-\dot{R}dt)^2+r^2d\Omega^2.
	\end{align}
	In these coordinates, we can immediately identify the shift vector as
	\begin{align}
		N^{r}(t,r)=-\dot{R}, \label{eq:Nr-dot}
	\end{align}
	and we should note that the conserved quantities $\mathcal{E}(x)$ and $M(x)$ are now expressed as functions of $(t,r)$. 
	
	Our next step is to transform this metric to Schwarzschild coordinates $(t_s,r)$ in order to determine whether vacuum solutions exist. To achieve this, we must eliminate the cross terms, which requires imposing the condition $g_{t_s r}=0$ after the coordinate transformation. The new time coordinate is taken to be a function $t_s=t_s(t,r)$. Combining these requirements leads to the condition
	\begin{align}
		\partial_{r}t_s=\frac{N^{r}}{1+\mathcal{E}-(N^r)^2}\partial_{t}t_s.
	\end{align}
	This coordinate transformation leads to the metric
	\begin{align}
		ds^2=-\frac{1+\mathcal{E}-(N^{r})^2}{(1+\mathcal{E})(\partial_tt_s)^2}dt^2_s+\frac{1}{1+\mathcal{E}-(N^{r})^2}dr^2+r^2d\Omega^2.
	\end{align}
	We immediately observe that, at this stage, the metric is not uniquely determined, since it still contains the arbitrary LTB energy function $\mathcal{E}$. A nonconstant $\mathcal{E}(x)$ should be understood as part of the choice of reference congruence/slicing, rather than as an obstruction to the existence of a timelike Killing vector in the underlying vacuum spacetime.
	For the static reconstruction pursued here, we therefore choose a shell-independent energy representative, which removes this slicing ambiguity and allows the metric coefficients to depend only on $(M,r)$. To achieve this, we introduce the following ansatz
	\begin{align}
		(N^{r})^2=\mathcal{E}+1-f(M,r).
	\end{align}
	Using this and imposing the condition $(1+\mathcal{E})(\partial_{t}t_s)^2=1$ then leads to a metric of the form
	\begin{align}
		ds^2=-f(M,r)dt^2_s+\frac{dr^2}{f(M,r)}+r^2d\Omega^2.
	\end{align}
	This metric is static, admitting the Killing vector $\partial_{t_s}$, and is unique for a given $M=m$.
	
	We now proceed by combining Eq.~\eqref{eq:Nr-dot} with Eq.~\eqref{eq:Poisson} for $\dot{v}$, from which we obtain
	\begin{align}
		N^{r}=-\frac{1}{2r^2}\partial_{b}H^{\mathrm{AdS}}_{s},
	\end{align}
	which leads to the partial differential equation
	\begin{align}
		(\partial_{b}H^{\mathrm{AdS}}_s)^2=4r^4\left(1+\mathcal{E}-f(M,r)\right).
	\end{align}
	From this point onward we specialize to the AdS case discussed above. The
	relevant shell Hamiltonian is therefore $H_s^{\mathrm{AdS}}$, with
	\begin{align}
		H_s^{\mathrm{AdS}}=R^3H_0^{\mathrm{AdS}},\quad	H_0^{\mathrm{AdS}}=H_0^{(0)}+\frac{1}{L^2}.
	\end{align}
	Since the AdS contribution is independent of the shell momentum variable
	$b$, one has
	\begin{align}
		\partial_b H_0^{\mathrm{AdS}}=\partial_b H_0^{(0)} .
	\end{align}
	Using the factorization
	$H_s^{\mathrm{AdS}}=r^3H_0^{\mathrm{AdS}}$, this becomes
	\begin{align}
		\left(\partial_b H_0^{(0)}\right)^2=4\left(\frac{1-f(M,r)}{r^2}	+\mathcal R_1\right),	\label{eq:static-reconstruction-pde}
	\end{align}
	where we used $\mathcal R_1=\mathcal E/r^2$. Equivalently,
	\begin{align}
		\psi:=\frac{1-f(M,r)}{r^2}=\frac14\left(\partial_b H_0^{(0)}\right)^2-\mathcal R_1.\label{eq:psi-reconstruction}
	\end{align}
	The vacuum condition is the conservation of the full AdS shell Hamiltonian,
	\begin{align}
		H_s^{\mathrm{AdS}}=2m .
	\end{align}
	Therefore
	\begin{align}
		H_0^{\mathrm{AdS}}=	\frac{2m}{R^3}, \quad	H_0^{(0)}=-\frac{1}{L^2}+\frac{2m}{R^3}.\label{eq:H0-vacuum-AdS}
	\end{align}
	
	For a generic polymerization, the right-hand side of Eq.~\eqref{eq:psi-reconstruction} depends independently on $b$ and $\mathcal R_1$. The Birkhoff-type polymerized vacuum sector considered here is defined by requiring that this combination depend on the local shell	data only through $H_0^{(0)}$. Equivalently,
	\begin{align}
		d\psi\wedge dH_0^{(0)}=0 .
	\end{align}
	The condition above is the mathematical expression of the Birkhoff-type requirement in the reduced shell language. For a generic polymerization, the Painleve--Gullstrand reconstruction gives a curvature variable $\psi$ that may depend separately on the local shell variables $b$ and $\mathcal R_1$. In that case, the value of the conserved shell Hamiltonian would not uniquely determine the metric function: different points in the local shell phase space could have the same value of $H_0^{(0)}$ but different values of $\psi$. The vacuum geometry would then contain additional shell data beyond the mass integration constant, and the Birkhoff-type property would be lost. The condition above removes this ambiguity. It states that $\psi$ is functionally dependent on $H_0^{(0)}$, or equivalently that $\psi$ is constant along the level sets of $H_0^{(0)}$. Therefore, once the vacuum condition fixes $H_s^{\rm AdS}=2m$, the metric function is determined only by $m$, the fixed AdS parameter, and the chosen polymerization function. 
	
	When this condition is satisfied, there exists a single-variable function
	$\tilde{f}$ such that
	\begin{align}
		\psi=\tilde{f}\left(H_0^{(0)}\right).\label{ft:H0}
	\end{align}
	Combining this with Eq.~\eqref{eq:H0-vacuum-AdS}, the static metric
	function becomes
	\begin{align}
		f(r)=1-r^2 \tilde{f}\left(-\frac{1}{L^2}+\frac{2Gm}{r^3}\right).
		\label{eq:f-LTB}
	\end{align}
	This equation is the central reconstruction formula for the asymptotically
	AdS polymerized vacuum geometries studied in the rest of this paper. The function $\tilde{f}$ specifies the effective polymerized model.
	
	For a nonlinear \(\tilde f\), the physical asymptotic AdS scale need not coincide with the bare parameter \(L\). At large radius,
	\begin{align}
		f(r)\simeq	1-r^2\tilde{f} \left(-\frac{1}{L^2}\right)-\tilde{f}'\left(-\frac{1}{L^2}\right)\frac{2Gm}{r}+\cdots,
	\end{align}
	so the effective AdS length and the correction to Newton's constant are determined by
	\begin{align}
		\tilde{f}\left(-\frac{1}{L^2}\right)=-\frac{1}{L_{\rm eff}^2}, \quad G_{\mathrm{eff}}=G\tilde{f}'\left(-\frac{1}{L^2}\right)\label{eq:GL-eff}
	\end{align}
	In the GR limit $\tilde{f}(X)=X$, one has $L_{\rm eff}=L$ and $G_{\mathrm{eff}}=G$, and the usual Schwarzschild--AdS form is recovered. For a general nonlinear $\tilde{f}$, the departure from GR should be understood as an effective modification of the vacuum gravitational dynamics in the reduced shell description. Explicit constructions of polymerized spherically symmetric constraints, their compatible LTB reductions, and related covariant realizations have been developed in Refs.~\cite{Han:2022rsx,GLRSW:25,GL:25,LS:26}. Through Eq.~\eqref{ft:H0}, $\tilde{f}$ changes the relation between the reconstructed curvature variable and the conserved gravitational shell Hamiltonian, thereby modifying the static vacuum geometry and the effective couplings appearing in Eq.~\eqref{eq:GL-eff}. No physical matter source is introduced: as discussed above, the dust field serves only as a relational clock and its energy density vanishes in the vacuum sector. At this level, $\tilde{f}$ specifies the reduced spherically symmetric vacuum dynamics but does not uniquely determine an off-shell covariant action; distinct covariant completions may reproduce the same reduced sector while differing in their additional degrees of freedom and in their dynamics away from this sector.
	
We are primarily interested in regular solutions. It is therefore important to identify the conditions that the polymerization function $\tilde{f}$ must satisfy in order for the geometry to remain regular throughout the spacetime, and in particular near the center. Since the areal radius coordinate degenerates at the center, we write the metric on the punctured domain $r>0$ and characterize regularity of the center through the limit $r\to0^+$. We say that the center is curvature-regular if all scalar polynomial curvature invariants constructed from the Riemann tensor remain bounded in this limit. 
	
\begin{figure*}[!htbp]
	\centering
	
	\subfloat[Models I and II]{
		\includegraphics[width=0.45\textwidth]{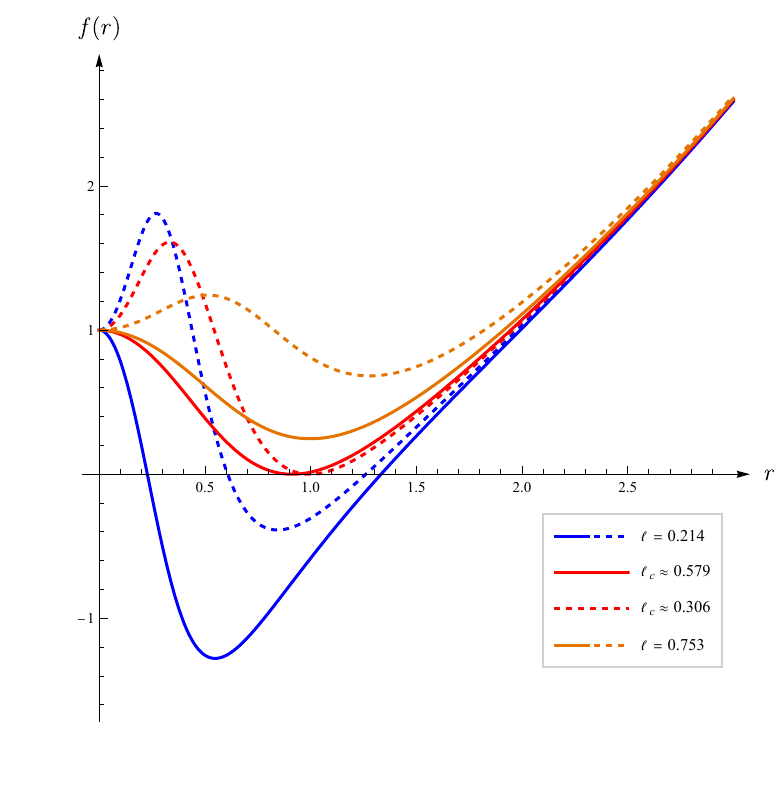}
		\label{fig:subfig:f12}
	}
	\hspace{0.03\textwidth}
	\subfloat[Models III and IV]{
		\includegraphics[width=0.45\textwidth]{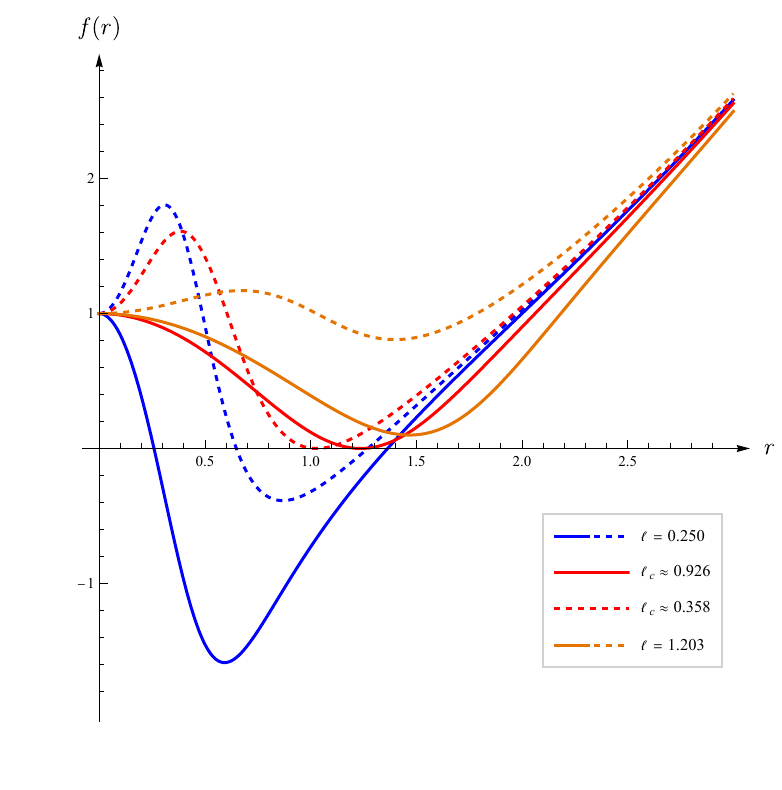}
		\label{fig:subfig:f34}
	}
	
	\vspace{0.35cm}
	
	\subfloat[Models V and VI]{
		\includegraphics[width=0.45\textwidth]{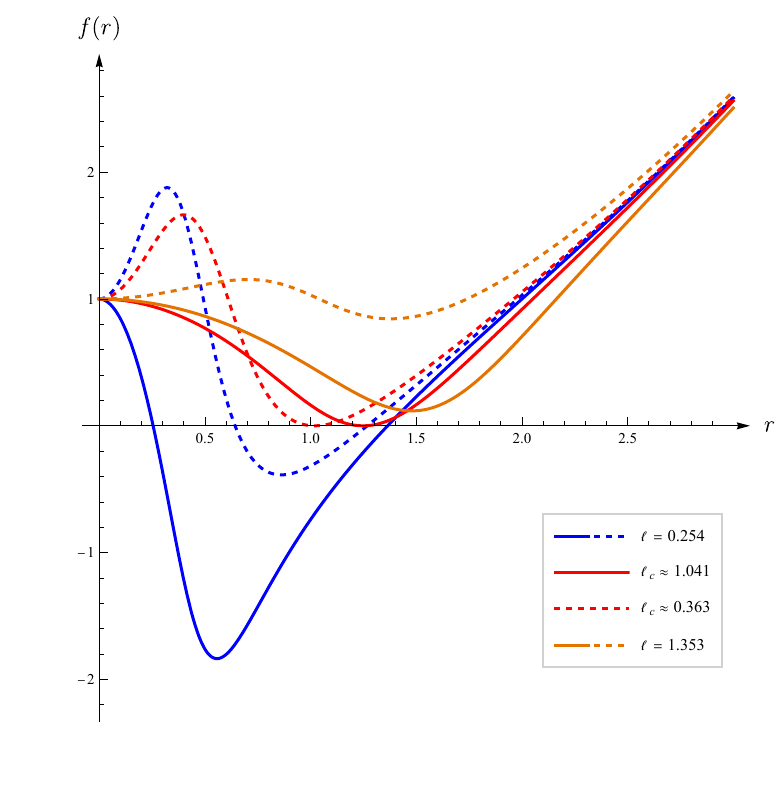}
		\label{fig:subfig:f56}
	}
	\hspace{0.03\textwidth}
	\subfloat[Models VII and VIII with $N=10$]{
		\includegraphics[width=0.45\textwidth]{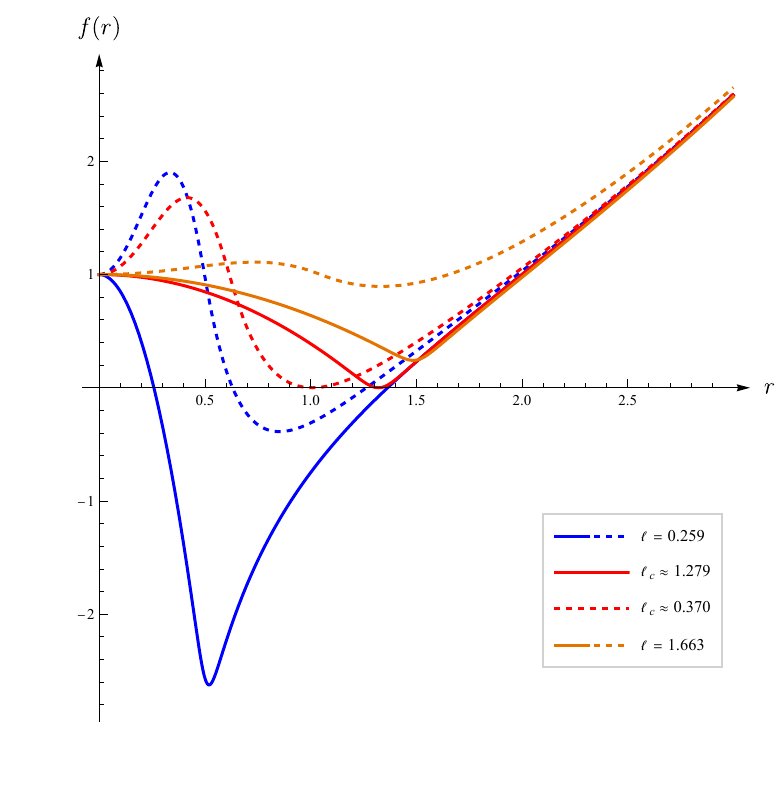}
		\label{fig:subfig:f78}
	}
	
	\caption{\textbf{Metric functions.}  Representative polymerization functions from Table~\ref{tab:polymerization-models} and their corresponding metric functions, with $u(r)=-1/L^2+2Gm/r^3$. Solid lines denote the odd models, which possess a dS core, while dashed lines denote the corresponding AdS core models. Three representative regimes are shown, depending on the value of the minimal length scale: an RBH with two horizons, corresponding to an outer and an inner horizon; an extremal black hole at the critical value $\ell=\ell_c$, where the two horizons merge; and a horizonless configuration for which no trapped region forms. The mass is fixed to $m=1$, while the effective Newton constant and effective AdS length are set to $G_{\mathrm{eff}}=1$ and $L_{\mathrm{eff}}=2$, respectively. The AdS core entries assume $\epsilon>0$; when $\epsilon=0$, these models reduce to their dS core counterparts. To make the distinction between dS core and AdS core behavior as transparent as possible, we set $\epsilon=1$.}
	\label{fig:f}
\end{figure*}

	\begin{theorem} \textbf{Curvature regularity criterion for polymerized static metrics:}
		\label{thm:polymerized-regularity}
		Consider a four-dimensional static, spherically symmetric metric written on the punctured areal-radius domain $\mathcal M^\circ=\{r>0\}$ as
		\begin{align}
			ds^2=
			-f(r)dt^2+\frac{dr^2}{f(r)}	+r^2d\Omega^2 , \quad f(r)=1-r^2\tilde f(u(r)).\label{eq:sss-metric}
		\end{align}
		Assume that $u\in C^2((0,\infty))$ and $\tilde f\in C^2(u((0,\infty)))$, and define
		\begin{align}
			F(r):=\tilde f(u(r)).
		\end{align}
		The center is curvature-regular, in the sense that all scalar polynomial curvature invariants constructed from the Riemann tensor remain bounded as $r\to0^+$, if and only if
		\begin{align}
			F(r),	\qquad rF'(r), \qquad r^2F''(r)	\label{eq:center-regularity-F}
		\end{align}
		remain bounded as $r\to0^+$. Equivalently, the criterion can be written as the boundedness of
		\begin{align}
			\tilde f(u(r)), &\qquad r\,\tilde f'(u(r))u'(r),	\\ &r^2
			\left[\tilde f''(u(r))\left(u'(r)\right)^2+\tilde f'(u(r))u''(r)\right], \label{eq:center-regularity-u}
		\end{align}
		in the same limit.
	\end{theorem}
	
	\begin{proof}
		The proof is local and follows by examining the Riemann tensor in an orthonormal frame. For a metric of the form \eqref{eq:sss-metric}, the independent curvature scales are\footnote{
			These combinations are obtained by evaluating the Riemann tensor in the local orthonormal coframe
			$e^{\hat t}=\sqrt{f}\,dt$,
			$e^{\hat r}=dr/\sqrt{f}$,
			$e^{\hat\theta}=r\,d\theta$,
			and
			$e^{\hat\phi}=r\sin\theta\,d\phi$.
			In this frame, the Riemann tensor acts naturally on two-planes, or equivalently on bivectors. The relevant two-planes are
			$\hat t\wedge\hat r$,
			$\hat t\wedge\hat\theta$,
			$\hat t\wedge\hat\phi$,
			$\hat r\wedge\hat\theta$,
			$\hat r\wedge\hat\phi$,
			and
			$\hat\theta\wedge\hat\phi$.
			Spherical symmetry identifies the two angular directions, while the gauge$g_{tt}g_{rr}=-1$ implies that the time-angular and radial-angular sectional curvatures have the same magnitude and opposite sign. With the Riemann-sign convention used here, the nonzero independent components are
			\begin{align}
				R_{\hat t\hat r\hat t\hat r}
				&=
				-\frac{1}{2}f''(r),
				\\
				R_{\hat t\hat\theta\hat t\hat\theta}
				=
				R_{\hat t\hat\phi\hat t\hat\phi}
				&=
				-\frac{1}{2r}f'(r),
				\\
				R_{\hat r\hat\theta\hat r\hat\theta}
				=
				R_{\hat r\hat\phi\hat r\hat\phi}
				&=
				\frac{1}{2r}f'(r),
				\\
				R_{\hat\theta\hat\phi\hat\theta\hat\phi}
				&=
				\frac{1-f(r)}{r^2}.
			\end{align}
			Thus the curvature is controlled by the three magnitudes
			$f''(r)$, $f'(r)/r$, and $(1-f(r))/r^2$. They may be viewed as the eigenvalue-like entries of the Riemann tensor regarded as an operator on two-forms. Since scalar polynomial curvature invariants are contractions of products of orthonormal-frame Riemann components, boundedness of these three curvature scales is equivalent, for this metric, to boundedness of all scalar polynomial curvature invariants.
		}
		\begin{align}
			f''(r),
			\qquad
			\frac{f'(r)}{r},
			\qquad
			\frac{1-f(r)}{r^2}.
			\label{eq:independent-curvature-scales}
		\end{align}
		Therefore the center is curvature-regular if and only if the three quantities in Eq.~\eqref{eq:independent-curvature-scales} remain bounded as $r\to0^+$.
		
		Using the global form \eqref{eq:sss-metric}, we write
		\begin{align}
			f(r)=1-r^2F(r),
			\qquad
			F(r)=\tilde f(u(r)).
		\end{align}
		A direct differentiation gives
		\begin{align}
			\frac{1-f(r)}{r^2}
			&=
			F(r),
			\label{eq:curv-scale-1}
			\\
			\frac{f'(r)}{r}
			&=
			-2F(r)-rF'(r),
			\label{eq:curv-scale-2}
			\\
			f''(r)
			&=
			-2F(r)-4rF'(r)-r^2F''(r).
			\label{eq:curv-scale-3}
		\end{align}
		
		If $F(r)$, $rF'(r)$, and $r^2F''(r)$ are bounded as $r\to0^+$, then Eqs.~\eqref{eq:curv-scale-1}--\eqref{eq:curv-scale-3} show that all three curvature scales are bounded. This proves sufficiency.
		
		Conversely, suppose that the center is curvature-regular. Then Eq.~\eqref{eq:curv-scale-1} implies that $F(r)$ is bounded. Equation~\eqref{eq:curv-scale-2}, together with boundedness of $F(r)$, implies boundedness of $rF'(r)$. Finally, Eq.~\eqref{eq:curv-scale-3}, together with boundedness of $F(r)$ and $rF'(r)$, implies boundedness of $r^2F''(r)$. This proves necessity.
		
		Hence the center is curvature-regular if and only if the three quantities in Eq.~\eqref{eq:center-regularity-F} remain bounded as $r\to0^+$. Since
		\begin{align}
			F'(r)&=	\tilde f'(u(r))u'(r),\\
			F''(r)&=\tilde f''(u(r))\left(u'(r)\right)^2+\tilde f'(u(r))u''(r),
		\end{align}
		the equivalent form \eqref{eq:center-regularity-u} follows immediately. This completes the proof.
	\end{proof}
	
	\begin{corollary}\textbf{Limiting core geometry:}
		\label{cor:limiting-core-geometry}
		Assume the conditions of Theorem~\ref{thm:polymerized-regularity}. If the limits
		\begin{align}
			\lim_{r\to0^+}F(r)=F_0,
			\quad
			\lim_{r\to0^+}rF'(r)=F_1,
			\quad
			\lim_{r\to0^+}r^2F''(r)=F_2
		\end{align}
		exist, then the Ricci scalar and Kretschmann scalar admit the finite limits
		\begin{align}
			\lim_{r\to0^+}R&=12F_0+8F_1+F_2,\\
			\lim_{r\to0^+}K&=\left(2F_0+4F_1+F_2\right)^2+4\left(2F_0+F_1\right)^2+4F_0^2 .
		\end{align}
		In particular, if $F_1=F_2=0$, then
		\begin{align}
			f(r)=1-F_0r^2+\mathcal{O}(r^{2+a}),\qquad r\to0^+,
		\end{align}
		with $a>0$, and the center is locally maximally symmetric at leading order, with
		\begin{align}
			R\to12F_0,\qquad K\to24F_0^2 .
		\end{align}
		Thus $F_0>0$ corresponds to a dS core, $F_0<0$ to an AdS core, and $F_0=0$ to a Minkowski core at leading order.
	\end{corollary}
	
	\begin{proof}[Proof of Corollary~\ref{cor:limiting-core-geometry}]
		For the metric \eqref{eq:sss-metric}, the Ricci scalar and Kretschmann scalar are
		\begin{align}
			R&=-f''(r)
			-\frac{4}{r}f'(r)+\frac{2}{r^2}\left[1-f(r)\right],\\
			K&=\left[f''(r)\right]^2+\frac{4}{r^2}\left[f'(r)\right]^2+\frac{4}{r^4}\left[1-f(r)\right]^2 .
		\end{align}
		Substituting $f(r)=1-r^2F(r)$ gives
		\begin{align}
			R&=12F(r)+8rF'(r)+r^2F''(r),\\
			K&=\left[2F(r)+4rF'(r)+r^2F''(r)\right]^2\\&+4\left[2F(r)+rF'(r)\right]^2
		    +4F(r)^2.
		\end{align}
		Taking the limits stated in the corollary gives the announced expressions. If $F_1=F_2=0$, then $F(r)=F_0+\mathcal{O}(r^a)$, and therefore
		\begin{align}
			f(r)=1-r^2F(r)=1-F_0r^2+\mathcal{O}(r^{2+a}).
		\end{align}
		The limiting values of $R$ and $K$ then reduce to $12F_0$ and $24F_0^2$, respectively.
	\end{proof}
	Although our construction differs substantially from quasitopological gravity, there are useful similarities at the level of the reduced static solutions. Making this comparison helps clarify the role of the reconstruction function in our framework and will also be useful later when comparing the thermodynamic properties of our four-dimensional solutions with the higher-dimensional solutions of Ref.~\cite{HKMS:25}. 
	
	It is useful to distinguish two different notions of invertibility. In quasitopological gravity, the reduced field equation typically takes the form
	\begin{align}
	h_f(\psi)=S(r),
	\end{align}
	so the metric function is obtained through the inverse map
	\begin{align}
	\psi=h_f^{-1}(S(r)).
	\end{align}
	Thus monotonicity of \(h_f\) on the domain sampled by \(S(r)\) is a natural sufficient condition for a single-valued and nonsingular metric reconstruction.
	
	In the polymerized construction used here, the Birkhoff-type condition instead implies
	\begin{align}
	\psi=\tilde f(H_0^{(0)}),
	\end{align}
	and therefore the static metric is reconstructed directly as
	\begin{align}
	f(r)=1-r^2\tilde f(u(r)),
	\qquad
	u(r)=-\frac{1}{L^2}+\frac{2Gm}{r^3}.
    \end{align}
	At this level, the metric requires \(\tilde f(u(r))\), rather than \(\tilde f^{-1}\), to be finite and sufficiently smooth on the physical interval sampled by \(u(r)\). Hence global invertibility of \(\tilde f\) is not a prerequisite for the existence of the static geometry.
	
	Nevertheless, invertibility and branch choices re-enter in two related ways. First, reconstructing the underlying physical Hamiltonian \(H_0(b,R_1)\) from a chosen polymerization function generally leads to branchwise solutions, and the physical Hamiltonian must be selected on appropriate monotonic segments \cite{LS:26}. Second, in the thermodynamic analysis, the horizon condition
	\begin{align}
	\tilde f(u_h)=\frac{1}{r_h^2},
	\end{align}
	is often rewritten as
	\begin{align}
	u_h=\tilde f^{-1}\left(\frac{1}{r_h^2}\right),
	\end{align}
	where $u_{h}=u(r_h)$, and $r_h$ denotes the outer horizon radius.
	
	This step requires choosing an inverse branch. To avoid ambiguities in the present thermodynamic analysis, we restrict attention to parameter ranges and horizon branches on which the relevant inverse is single-valued.

	\begin{table*}[t]
		\centering
		\renewcommand{\arraystretch}{1.8}
		\setlength{\tabcolsep}{8pt}
		\begin{tabular}{|c||c||c||c|}
			\hline
			\rule{0pt}{3.8ex}
			Model 
			& $\tilde{f}(\mathcal{X})$ 
			& $f(r)$
			& Core geometry
			\rule[-1.6ex]{0pt}{0pt}
			\\
			\hline\hline
			
			\rule{0pt}{4.8ex}
			I 
			& $\displaystyle\frac{\mathcal{X}}{1+\ell^2\mathcal{X}}$
			& $\displaystyle 1-r^2\frac{u(r)}{1+\ell^2u(r)}$
			& dS
			\rule[-2.0ex]{0pt}{0pt}
			\\
			\hline
			
			\rule{0pt}{5.2ex}
			II
			& $\displaystyle\frac{\mathcal{X}}{1+\ell^2\mathcal{X}}\frac{1-\epsilon\ell^2\mathcal{X}}{1+\epsilon\ell^2\mathcal{X}}$
			& $\displaystyle 1-r^2\frac{u(r)}{1+\ell^2u(r)}\frac{1-\epsilon\ell^2u(r)}{1+\epsilon\ell^2u(r)}$
			& AdS
			\rule[-2.2ex]{0pt}{0pt}
			\\
			\hline
			
			\rule{0pt}{5.2ex}
			III
			& $\displaystyle \frac{-1+\sqrt{1+4\ell^4\mathcal{X}^2}}{2\ell^4\mathcal{X}}$
			& $\displaystyle 1-r^2\frac{-1+\sqrt{1+4\ell^4u(r)^2}}{2\ell^4u(r)}$
			& dS
			\rule[-2.2ex]{0pt}{0pt}
			\\
			\hline
			
			\rule{0pt}{5.2ex}
			IV
			& $\displaystyle \frac{-1+\sqrt{1+4\ell^4\mathcal{X}^2}}{2\ell^4\mathcal{X}}\frac{1-\epsilon\ell^2\mathcal{X}}{1+\epsilon\ell^2\mathcal{X}}$
			& $\displaystyle 1-r^2\frac{-1+\sqrt{1+4\ell^4u(r)^2}}{2\ell^4u(r)}\frac{1-\epsilon\ell^2u(r)}{1+\epsilon\ell^2u(r)}$
			& AdS
			\rule[-2.2ex]{0pt}{0pt}
			\\
			\hline
			
			\rule{0pt}{4.8ex}
			V
			& $\displaystyle \frac{\mathcal{X}}{\sqrt{1+\ell^4\mathcal{X}^2}}$
			& $\displaystyle 1-r^2\frac{u(r)}{\sqrt{1+\ell^4u(r)^2}}$
			& dS
			\rule[-2.0ex]{0pt}{0pt}
			\\
			\hline
			
			\rule{0pt}{5.2ex}
			VI
			& $\displaystyle \frac{\mathcal{X}}{\sqrt{1+\ell^4\mathcal{X}^2}}\frac{1-\epsilon\ell^2\mathcal{X}}{1+\epsilon\ell^2\mathcal{X}}$
			& $\displaystyle 1-r^2\frac{u(r)}{\sqrt{1+\ell^4u(r)^2}}\frac{1-\epsilon\ell^2u(r)}{1+\epsilon\ell^2u(r)}$
			& AdS
			\rule[-2.2ex]{0pt}{0pt}
			\\
			\hline
			
			\rule{0pt}{5.2ex}
			VII
			& $\displaystyle \frac{\mathcal{X}}{(1+\ell^{2\mathrm{N}}\mathcal{X}^{\mathrm{N}})^{1/\mathrm{N}}}$
			& $\displaystyle 1-r^2\frac{u(r)}{\left(1+\ell^{2\mathrm{N}}u(r)^{\mathrm{N}}\right)^{1/\mathrm{N}}}$
			& dS
			\rule[-2.2ex]{0pt}{0pt}
			\\
			\hline
			
			\rule{0pt}{5.2ex}
			VIII
			& $\displaystyle \frac{\mathcal{X}}{(1+\ell^{2\mathrm{N}}\mathcal{X}^{\mathrm{N}})^{1/\mathrm{N}}}\frac{1-\epsilon\ell^2\mathcal{X}}{1+\epsilon\ell^2\mathcal{X}}$
			& $\displaystyle 1-r^2\frac{u(r)}{\left(1+\ell^{2\mathrm{N}}u(r)^{\mathrm{N}}\right)^{1/\mathrm{N}}}\frac{1-\epsilon\ell^2u(r)}{1+\epsilon\ell^2u(r)}$
			& AdS
			\rule[-2.2ex]{0pt}{0pt}
			\\
			\hline
		\end{tabular}
		\caption{\textbf{Models: Metric functions and  core geometry.} Representative polymerization functions and their corresponding metric functions. Here $u(r)=-1/L^2+2Gm/r^3$. The AdS-core entries assume $\epsilon>0$; for $\epsilon=0$ the corresponding models reduce to their dS core counterparts.}
		\label{tab:polymerization-models}
	\end{table*}

\begin{table*}[t]
	\centering
	\renewcommand{\arraystretch}{1.8}
	\setlength{\tabcolsep}{8pt}
	\begin{tabular}{|c||c||c||c||c|}
		\hline
		\rule{0pt}{3.8ex}
		Model 
		& $\tilde{f}(\mathcal{X})$ 
		& $L^2_{\mathrm{eff}}/L^2$
		& $G_{\mathrm{eff}}/G$
		& Parameter range
		\rule[-1.6ex]{0pt}{0pt}
		\\
		\hline\hline
		
		\rule{0pt}{4.8ex}
		I 
		& $\displaystyle\frac{\mathcal{X}}{1+\ell^2\mathcal{X}}$
		& $1-\bar{\ell}^2$
		& $\displaystyle \frac{1}{\left(1-\bar{\ell}^2\right)^2}$
		& $L^2>\ell^2$
		\rule[-2.0ex]{0pt}{0pt}
		\\
		\hline
		
		\rule{0pt}{5.2ex}
		II
		& $\displaystyle\frac{\mathcal{X}}{1+\ell^2\mathcal{X}}\chi(r)$
		& $\displaystyle \frac{\left(1-\bar{\ell}^2\right)\left(1-\epsilon\bar{\ell}^2\right)}{1+\epsilon\bar{\ell}^2}$
		& $\displaystyle \frac{1+2\epsilon\bar{\ell}^2-\left(\epsilon^2+2\epsilon\right)\bar{\ell}^4}{\left(1-\bar{\ell}^2\right)^2\left(1-\epsilon\bar{\ell}^2\right)^2}$
		& $L^2>\ell^2$
		\rule[-2.2ex]{0pt}{0pt}
		\\
		\hline
		
		\rule{0pt}{5.2ex}
		III
		& $\displaystyle \frac{-1+\sqrt{1+4\ell^4\mathcal{X}^2}}{2\ell^4\mathcal{X}}$
		& $\displaystyle \frac{2\bar{\ell}^4}{\sqrt{1+4\bar{\ell}^4}-1}$
		& $\displaystyle \frac{-1+\sqrt{1+4\bar{\ell}^4}}{2\bar{\ell}^4\sqrt{1+4\bar{\ell}^4}}$
		& ---
		\rule[-2.2ex]{0pt}{0pt}
		\\
		\hline
		
		\rule{0pt}{5.2ex}
		IV
		& $\displaystyle \frac{-1+\sqrt{1+4\ell^4\mathcal{X}^2}}{2\ell^4\mathcal{X}}\chi(r)$
		& $\displaystyle \frac{2\bar{\ell}^4\left(1-\epsilon\bar{\ell}^2\right)}{\left(1+\epsilon\bar{\ell}^2\right)\left(-1+\sqrt{1+4\bar{\ell}^4}\right)}$
		& $\displaystyle \frac{Y_{\ast}-\epsilon^2\bar{\ell}^4Y_{\ast}-3\epsilon\bar{\ell}^2(4\bar{\ell}^4+Y_{\ast})}{2\bar{\ell}^4(1-\epsilon^2\bar{\ell}^2)(1-Y_{\ast})}$
		& $L^2>\epsilon \ell^2$
		\rule[-2.2ex]{0pt}{0pt}
		\\
		\hline
		
		\rule{0pt}{4.8ex}
		V
		& $\displaystyle \frac{\mathcal{X}}{\sqrt{1+\ell^4\mathcal{X}^2}}$
		& $\displaystyle \sqrt{1+\bar{\ell}^4}$
		& $\displaystyle \frac{1}{(1+\bar{\ell}^4)^{3/2}}$
		& ---
		\rule[-2.0ex]{0pt}{0pt}
		\\
		\hline
		
		\rule{0pt}{5.2ex}
		VI
		& $\displaystyle \frac{\mathcal{X}}{\sqrt{1+\ell^4\mathcal{X}^2}}\chi(r)$
		& $\displaystyle \frac{(1-\epsilon\bar{\ell}^2)\sqrt{1+\bar{\ell}^4}}{1+\epsilon\bar{\ell}^2}$
		& $\displaystyle \frac{1+2\epsilon(\bar{\ell}^6+\bar{\ell}^2)-\epsilon^2\bar{\ell}^4}{(1-\epsilon^2\bar{\ell}^2)(1+\bar{\ell}^4)^{3/2}}$
		& $L^2>\epsilon \ell^2$
		\rule[-2.2ex]{0pt}{0pt}
		\\
		\hline
		
		\rule{0pt}{5.2ex}
		VII
		& $\displaystyle \frac{\mathcal{X}}{(1+\ell^{2\mathrm{N}}\mathcal{X}^{\mathrm{N}})^{1/\mathrm{N}}}$
		& $\displaystyle \left(1+(-\bar{\ell}^{2})^{\mathrm{N}}\right)^{1/\mathrm{N}}$
		& $\displaystyle \left(1+(-\bar{\ell}^2)^{\mathrm{N}}\right)^{-\frac{\mathrm{N}+1}{\mathrm{N}}}$
		& $\begin{array}{c}
			\mathrm{N}\ \mathrm{even}: \text{---}\\
			\mathrm{N}\ \mathrm{odd}: L^2>\ell^2
		\end{array}$
		\rule[-2.2ex]{0pt}{0pt}
		\\
		\hline
		
		\rule{0pt}{5.2ex}
		VIII
		& $\displaystyle \frac{\mathcal{X}}{(1+\ell^{2\mathrm{N}}\mathcal{X}^{\mathrm{N}})^{1/\mathrm{N}}}\chi(r)$
		& $\displaystyle \frac{(1-\epsilon^2\bar{\ell}^2)\left(1+(-\bar{\ell}^{2})^{\mathrm{N}}\right)^{1/\mathrm{N}}}{1+\epsilon\bar{\ell}^2}$
		& $\displaystyle \frac{(1+2\epsilon\bar{\ell}^2)(1+(-\bar{\ell}^2)^{\mathrm{N}}-\epsilon^2\bar{\ell}^4)}{(1-\epsilon\bar{\ell}^2)\left(1+(-\bar{\ell}^2)^{\mathrm{N}}\right)^{\frac{\mathrm{N}+1}{\mathrm{N}}}}$
		& $\begin{array}{c}
			\mathrm{N}\ \mathrm{even}: L^2>\epsilon \ell^2\\
			\mathrm{N}\ \mathrm{odd}: L^2>\ell^2
		\end{array}$
		\rule[-2.2ex]{0pt}{0pt}
		\\
		\hline
	\end{tabular}
	\caption{\textbf{Models: Effective couplings and parameter ranges.}
		Representative polymerization functions and their corresponding effective AdS length and effective Newton coupling. To simplify the expressions, we introduce the dimensionless parameter $\bar{\ell}=\ell/L$, the auxiliary quantity $Y_{\ast}=1-\sqrt{1+4\bar{\ell}^{\,2}}$, and $\chi(r)=\frac{1-\epsilon\ell^2\mathcal{X}}{1+\epsilon\ell^2\mathcal{X}}$, with $\mathcal{X}=\frac{2m}{r^3}$. The last column gives the parameter ranges for which the geometries are fully regular and admit AdS asymptotics.}
	\label{tab:polymerization-effective}
\end{table*}

In the quasitopological construction, regularity is often ensured by requiring the theory function $h_f$ to diverge at a finite value of $\psi$, so that the inverse $h_f^{-1}(S)$ approaches a finite constant as $S\to\infty$. In the present polymerized reconstruction the role of $h_f^{-1}$ is played directly by $\tilde f$. Thus the corresponding regularity requirement is not a divergence of $\tilde f$, but rather saturation of $\tilde f(u)$ to a finite value, with sufficiently mild derivatives, as $u\to+\infty$.

The purpose of our construction is to generate polymerized vacuum regular ultracompact object geometries with asymptotically AdS behavior. This is achieved by choosing the polymerization function $\tilde{f}$ so that it satisfies the conditions stated in Theorem~\ref{thm:polymerized-regularity}. The resulting geometries are unique in the following Birkhoff-type sense: once the reconstruction function \(\tilde f\) and the bare AdS length \(L\) are specified, the static spacetime is determined by a single integration constant, identified with the mass parameter. Regular ultracompact-object geometries are most often constructed by replacing the central singularity with a  dS core, as happens in well-known examples such as the Bardeen \cite{B:68}, Hayward \cite{H:06}, and Dymnikova spacetimes \cite{D:92}. This is not, however, the only possible mechanism for regularization. One may also obtain nonsingular geometries whose central region is AdS-like. Regardless of whether the core is dS or AdS, the avoidance of a curvature singularity requires a departure from the assumptions of the classical singularity theorems \cite{Bambi:book:23}. Equivalently, a regular geometry cannot obey all standard classical energy conditions everywhere; which condition fails, and where the violation takes place, depends on the specific model under consideration. A dS core is associated with positive energy density and negative pressure, and therefore commonly leads to a violation of the strong energy condition in the central region. An exact AdS-type core has the opposite effective behavior: the energy density is negative while the pressure is positive. 

There is also a physical motivation for considering AdS-like interiors. As a compact object is compressed toward the Buchdahl bound \cite{B:59}, the pressure inside the object increases and can become arbitrarily large and positive. Such extreme pressures may enhance quantum effects, whose backreaction can be described by an effective energy-momentum tensor with negative energy density near the origin \cite{ABRG:22,ABRG:24,RT:23,R:23}. This provides a natural setting in which an AdS-core regularization can arise. In addition, recent work has shown that regular metrics originally formulated with dS cores can be consistently deformed so that their central behavior becomes AdS-like instead \cite{ALNV:25}.

AdS core interiors have also appeared in other RBH constructions, including quasitopological gravity and models based on nonlinear electrodynamics. The mechanism, however, is different from the one considered here. In nonlinear electrodynamics models, the change from a dS-like to an AdS-like core is tied to the electromagnetic sector: the charge, or equivalently the electrostatic self-energy, introduces an additional physical scale and can change the sign of the effective near-center curvature \cite{HKMS:25}. Consequently, the thermodynamic phase structure in such models depends not only on regularity, but also on the extra matter charge and its conjugate thermodynamic variable. By contrast, the polymerized solutions studied here are vacuum regularizations. No electromagnetic charge is introduced; the dS-core and AdS core branches are generated by different choices of the reconstruction function within the same vacuum framework. This makes the comparison below a test of how the vacuum resolution of the singularity itself affects the thermodynamics. Polymerized vacuum solutions with AdS cores were constructed in Ref.~\cite{LS:26}\footnote{Given the similarity of the reduced static equations in spherical symmetry, one may expect both polynomial and non-polynomial quasitopological gravity theories to accommodate vacuum RBH solutions in which the singularity is resolved by an AdS core. For example, the constructions of Refs.~\cite{BCH:25} and \cite{BCR:26} provide natural frameworks in which to investigate such solutions.} .

We briefly review the procedure for constructing AdS core metrics from their dS core counterparts. A detailed discussion can be found in Ref.~\cite{ALNV:25}, and its application to the polymerized case in Ref.~\cite{LS:26}. The basic idea is to modify the Misner-Sharp mass in such a way that the coefficient of the $\mathcal{O}(r^2)$ term in the near-center expansion changes sign, thereby converting the dS core into an AdS core. At the same time, the modification must preserve the desired large radius behavior. In the asymptotically flat case this means retaining the Minkowski limit together with the standard Schwarzschild falloff. In the present setting the large radius limit is instead AdS: the metric must approach a Schwarzschild--AdS form with the appropriate effective AdS length and effective Newton coupling determined by the asymptotic branch of the reconstruction function.

Following Ref.~\cite{ALNV:25}, a regular geometry with a dS core can be converted into one with an AdS core by modifying the effective mass profile. Let the original dS core metric be written as

\begin{align}
	f(r)=1-\frac{2M_{\mathrm{dS}}(r)}{r}.
\end{align}

For such geometries, the Misner-Sharp mass is positive close to the origin, reflecting the dS-like central behavior. By contrast, an AdS-like core requires the effective Misner-Sharp mass to become negative in the central region. To implement this change while keeping the large distance behavior unchanged, we multiply the original mass function by an interpolating factor $\chi(r)$, so that
\begin{align}
	f(r)=1-\frac{2M_{\mathrm{dS}}(r)\chi(r)}{r}.
	\label{eq:fmod}
\end{align}
The role of $\chi(r)$ is to reverse the sign of the mass function near the center while leaving the asymptotic Schwarzschild form intact. Therefore it must obey
\begin{align}
	\lim_{r\to\infty}\chi(r)=1,
	\qquad
	\lim_{r\to0}\chi(r)=-1.
	\label{eq:chi-limits}
\end{align}
Many interpolating functions can satisfy these two conditions. In this work we use the form proposed in Ref.~\cite{ALNV:25}, which is technically convenient for the applications considered below:
\begin{align}
	\chi(r)=
	\frac{r^n-\epsilon \kappa^{\,n-p}\ell^p}
	{r^n+\epsilon \kappa^{\,n-p}\ell^p},
	\qquad
	n\geqslant1,\quad p\geqslant1.
	\label{eq:chi}
\end{align}
Here $\kappa$ is a positive length scale, while $\ell$ denotes the regularization scale. The latter is usually interpreted as the length at which quantum gravity corrections to the classical Einstein equations become relevant \cite{F:16}. The rational expression in Eq.~\eqref{eq:chi} should be viewed as one convenient representative rather than a unique choice. Other functions with the same limiting behavior in Eq.~\eqref{eq:chi-limits} would lead to different AdS core completions of the same original dS core geometry.

We also stress that $\ell$ need not be fixed in advance to the Planck length $\ell_p$ \cite{COS:22,CLMMOS:23}\footnote{This scale need not be universal or constant. In particular, studies of black hole interiors including backreaction effects suggest that the minimal length $\ell$ may evolve dynamically \cite{BBCRG:21,BBCRG:22,CRFLV:23}. This provides motivation for allowing $\ell$ to vary when studying black hole thermodynamics \cite{SS:24,S:24} and possible observational signatures \cite{ms:24}.}. Finally, the parameter \(\epsilon\) in Eq.~\eqref{eq:chi} is dimensionless. In this work we restrict to
\begin{align}
0\leqslant \epsilon \leqslant 1 .
\end{align}
The endpoint \(\epsilon=0\) removes the deformation and restores the original dS core metric, while \(0<\epsilon\leqslant1\) gives AdS core deformations of increasing strength. The lower bound ensures that the denominator of \(\chi(r)\) does not introduce an additional pole at positive radius. The upper bound fixes a convenient representative range of deformations; within this range the admissible AdS branches used below are selected by the regularity and asymptotic conditions summarized in Table~\ref{tab:polymerization-effective}. In the plots and explicit thermodynamic comparisons we use the representative choice \(\epsilon=1\).

This shows that, once a particular polymerization function $\tilde f(\mathcal X)$ has been specified, it can be deformed through the interpolating function $\chi(r)$ in order to obtain the AdS core counterpart of the original dS core polymerized vacuum solution. This is the construction adopted in the present paper before turning to the thermodynamic analysis. Representative models are collected in Table~\ref{tab:polymerization-models}, while the corresponding metric functions for different values of the minimal length scale are displayed in Fig.~\ref{fig:f}. Odd models possess dS cores, whereas even models develop AdS cores when $\epsilon\neq 0$.

For fixed bare couplings $L$ and $G$, the AdS core deformation does not modify the bare AdS length or the bare Newton constant. Rather, it changes the branch selected by the reconstruction function $\tilde{f}$, and hence the effective asymptotic quantities extracted from the large-radius expansion. This is analogous to what occurs in quasitopological gravity, where a chosen resummation function determines an effective AdS length and an effective Newton coupling through the asymptotic branch of the solution \cite{HKMS:25}. In the present polymerized construction, these quantities are determined by Eq.~\eqref{eq:GL-eff}, and the resulting branch-dependent values for the models considered here are summarized in Table~\ref{tab:polymerization-effective}.

A well-defined asymptotically AdS expansion, together with the absence of additional singularities associated with pathological branches of $\tilde{f}$, imposes restrictions on the ratio between the regularization scale and the bare AdS length. These restrictions are displayed in the rightmost column of Table~\ref{tab:polymerization-effective}. In what follows, we focus on the RBH branch with two horizons, since this branch has nonzero Hawking temperature and therefore provides the relevant setting for the thermodynamic analysis.

\section{Thermodynamics}\label{sec:thermo}

There are several complementary ways to formulate black hole thermodynamics. Depending on the theory and the ensemble under consideration, the entropy may be obtained from the covariant phase-space construction of Iyer and Wald~\cite{IW:94}, the Euclidean path integral~\cite{GH:77}, conical-deficit methods~\cite{SU:94}, entanglement entropy arguments~\cite{S:11}, or microscopic/CFT techniques~\cite{C:99}. In the present work, we adopt a conservative semiclassical prescription adapted to the reduced effective description of the polymerized vacuum solutions. The temperature is fixed by the surface gravity at the outer horizon, while the entropy is defined thermodynamically by integrating the first law.

For a static black hole, regularity of the Euclidean section fixes the temperature to be
\begin{align}
	T=\frac{\kappa_{\rm s}}{2\pi},
\end{align}
where \(\kappa_{\rm s}\) is the surface gravity of the outer horizon. For the metric~\eqref{eq:metric}, this gives
\begin{align}
	T=\frac{1}{4\pi}\left[-2r_h\tilde f(u_h)+\frac{6Gm}{r_h^2}\tilde f'(u_h)\right],	\label{eq:temp}
\end{align}
where 
\begin{align}
u_h=-\frac{1}{L^2}+\frac{2Gm}{r_h^3}.
\end{align}

Here \(r_h\) denotes the outer horizon radius\footnote{When the horizon equation admits several real algebraic branches, we select the branch describing the outer black hole horizon. This branch has \(r_h>0\), has positive temperature on the thermodynamic branch, and reduces continuously to the Schwarzschild--AdS outer horizon as the regularization scale is sent to zero. Other branches are excluded when they lack this GR limit or do not describe the outer horizon relevant to the Hawking--Page comparison.}. The horizon condition \(f(r_h)=0\) implies
\begin{align}
	\tilde f(u_h)=\frac{1}{r_h^2},
	\qquad
	u_h=\tilde f^{-1}\left(\frac{1}{r_h^2}\right).
\end{align}
Using this relation, the conserved mass parameter of the reduced description can be written as a function of the horizon radius,
\begin{align}
	m(r_h) =\frac{r_h^3}{2G}\left[\frac{1}{L^2}+\tilde f^{-1}\left(\frac{1}{r_h^2}\right)\right].\label{eq:ADM-mass}
\end{align}
Substituting Eq.~\eqref{eq:ADM-mass} into Eq.~\eqref{eq:temp}, the temperature becomes
\begin{align}
	T(r_h)=\frac{1}{4\pi}\left[-\frac{2}{r_h}+3r_h\left(\frac{1}{L^2}+u_h\right)\tilde f'(u_h)\right].\label{eq:temp-simp}
\end{align}
Equations~\eqref{eq:ADM-mass} and~\eqref{eq:temp-simp} show explicitly that the thermodynamics is controlled by the reconstruction function \(\tilde f\).

A comment on the definition of temperature is useful at this point. In related effective theories, including covariant completions with shift-symmetric scalar sectors, alternative effective temperatures have sometimes been proposed~\cite{LHK:23}. Similar issues arise for RBHs sourced by nonlinear electrodynamics, where modified thermodynamic temperatures can appear once the constraints associated with the magnetic charge are included~\cite{SS:24,S:24}\footnote{A general discussion of black hole thermodynamics in the presence of constraints, with applications to several black hole spacetimes including RBHs, can be found in Ref.~\cite{MLS:25}.}. Such alternative definitions generally depend on additional information about the covariant completion, or on extra boundary terms in the Euclidean action. Since the polymerized model considered here is specified at the reduced Hamiltonian/effective level, we do not attempt to add completion-dependent contributions to the Euclidean action. Instead, we consistently use the surface-gravity temperature in the first law and in the phase structure analysis below.

We now turn to the entropy. In a fixed generally covariant completion, the entropy entering the first law should be given by the Iyer--Wald Noether charge, up to the standard ambiguities associated with boundary terms, total derivatives, and the normalization of the Noether charge. The same qualification applies to the conserved mass, which in a specified covariant theory should be computed using the corresponding covariant phase-space or Abbott--Deser--Tekin charge \cite{AD:82,DT:03}. The reduced polymerized model does not by itself select a unique off-shell covariant action. Different completions may agree on the reduced static solutions while differing by terms that affect both the Wald entropy and the normalization of the conserved charges. We therefore adopt a reduced thermodynamic prescription: the mass is identified with the conserved shell-Hamiltonian mass parameter, with its normalization fixed by the asymptotic metric, while the entropy is determined by integrating the first law. These quantities should be interpreted as the reduced thermodynamic mass and entropy that a covariant completion reproducing the same thermodynamic sector would have to recover.

The thermodynamic variables are most conveniently organized by treating the black hole solutions as a family labelled by \((r_h,P,\ell)\), where \(r_h\) is the outer horizon radius, \(P\) is the pressure associated with the bare cosmological constant, and \(\ell\) is the regularization scale. The pressure is defined by
\begin{align}
	P=\frac{3}{8\pi G L^2}.
\end{align}
Using the horizon condition, the mass may be written as
\begin{align}
	M(r_h,P,\ell)=\frac{4\pi}{3}P r_h^3+\frac{r_h^3}{2G}u_h(r_h,\ell),
	\quad
	\tilde f(u_h;\ell)=\frac{1}{r_h^2}.
	\label{eq:mass-extended}
\end{align}
Here $M$ denotes the conserved mass parameter of the reduced Hamiltonian description and is interpreted as enthalpy in the extended phase space \cite{KRT:09}. Its normalization is fixed by the asymptotic form of the metric,
	\begin{align}
	f(r)=1+\frac{r^{2}}{L_{\rm eff}^{2}}-\frac{2G_{\rm eff}M}{r}+\cdots ,
	\end{align}
	so that $M$ coincides with the vacuum shell integration constant $m$ introduced in Sec.~\ref{sec:LTB}. In a specified covariant completion, its identification with the corresponding ADM, Abbott--Deser--Tekin \cite{AD:82,DT:03}, or covariant phase-space charge must be verified using that theory’s action and boundary terms. Accordingly, the first law, free energy, and phase diagrams presented below should be understood within the reduced thermodynamic prescription adopted here.

The entropy is then determined from
\begin{align}
	\frac{\partial S}{\partial r_h}=\frac{1}{T}\left(\frac{\partial M}{\partial r_h}\right)_{P,\ell}.\label{eq:entropy-first-law}
\end{align}
Using Eqs.~\eqref{eq:temp-simp} and~\eqref{eq:mass-extended}, this gives the general expression
\begin{align}
	S=\frac{1}{G}\int dr_h\,\frac{2\pi r_h}{\tilde f'\!\left[\tilde f^{-1}\!\left(\frac{1}{r_h^2}\right)\right] }.
	\label{eq:entropy-general}
\end{align}
In the GR limit \(\tilde f(\mathcal X)=\mathcal X\), one has \(\tilde f'(\mathcal X)=1\), and Eq.~\eqref{eq:entropy-general} reduces to
\begin{align}
	S_{\rm GR}=\frac{1}{G}\int dr_h\,2\pi r_h=\frac{\pi r_h^2}{G}=\frac{A}{4G},
\end{align}
up to the usual normalization conventions for the gravitational coupling. Thus the Bekenstein--Hawking result is recovered in the classical limit.

With the entropy defined in this way, the extended first law takes the form
\begin{align}
	dM=T\,dS+V\,dP+\Psi\,d\ell .
	\label{eq:first-law-extended}
\end{align}
The quantities conjugate to \(P\) and \(\ell\) are
\begin{align}
	V=\left(\frac{\partial M}{\partial P}\right)_{S,\ell},
	\qquad
	\Psi=
	\left(\frac{\partial M}{\partial \ell}\right)_{r_h,P}-T\left(\frac{\partial S}{\partial \ell}\right)_{r_h}.
	\label{eq:conjugates}
\end{align}
Thus both the bare cosmological constant and the regularization scale are included in the thermodynamic phase space. In the phase structure analysis below, individual free-energy diagrams correspond to fixed choices of \(P\) and \(\ell\), but the underlying thermodynamic framework is the extended one described by Eq.~\eqref{eq:first-law-extended}.

In what follows, the thermodynamic pressure is associated with the bare cosmological constant entering the polymerized Hamiltonian, rather than with the effective cosmological constant read off from the asymptotic metric. This is the natural choice in the extended phase space: the bare cosmological constant is an independent coupling of the theory and can therefore be varied in the first law. By contrast, the effective cosmological constant is not an independent coupling. It is a derived quantity, obtained only after solving the vacuum relation for the chosen asymptotic branch. In general, the effective AdS radius depends on both the bare cosmological constant and the polymerization parameters. Treating the effective radius as the pressure would therefore mix an independent coupling with a derived quantity associated with the chosen asymptotic branch. This is the same distinction that appears in higher-curvature gravity, where the bare cosmological constant is varied as a coupling even though the physical AdS radius is determined only after solving the vacuum equation.

The Smarr relation follows from the scaling properties of the thermodynamic variables~\cite{Smarr:73}. While the first law gives the local variation of the mass, the Smarr formula gives the corresponding integrated scaling identity. In the present case, the regularization length \(\ell\) introduces an additional scale, and its conjugate contribution is required for homogeneous scaling. The relevant Smarr relation is
\begin{align}
	M=2TS+\ell\Psi-2PV .
	\label{eq:smarr}
\end{align}
This is consistent with the scaling
\begin{align}
	M&\to cM,
	&
	\ell&\to c\ell,
	&
	r_h&\to cr_h,
	&
	P&\to c^{-2}P,
	\nonumber\\
	S&\to c^2S,
	&
	T&\to c^{-1}T,
	&
	\Psi&\to \Psi,
	&
	V&\to c^3V .
\end{align}
Similar considerations apply in higher-curvature theories, where additional dimensionful couplings must also be included in the first law and Smarr relation~\cite{KRT:11,HBM:15,KRT:09}. For the models considered here, the relevant expressions for the thermodynamic conjugate potentials are given in Appendix~\ref{sec:app:thermo}.

Finally, we comment on the normalization of the entropy. Since the entropy is obtained by integrating the first law, it is fixed only up to a term independent of the horizon radius. Such a term is irrelevant for the fixed-\(\ell\) first law, but it is not irrelevant in the extended phase space, where the regularization scale is varied. In particular, logarithmic terms generated by the integration must be written with dimensionless arguments, as shown explicitly for the Hayward model in Appendix~\ref{sec:app:thermo}. We choose the integration constant so that the entropy is homogeneous under a common scaling of the length variables. With this prescription, the extended Smarr relation is satisfied. The same normalization also reproduces the standard Schwarzschild--AdS free-energy normalization in the GR limit: taking \(\ell\to0\) first and then \(r_h\to0\), equivalently \(m\to0\) after the GR limit has been taken, gives \(F=0\) for the thermal AdS background. This is the normalization used in the Hawking--Page analysis in Sec.~\ref{sec:PT}.

\begin{figure}[t]
	\centering
	\includegraphics[width=\columnwidth]{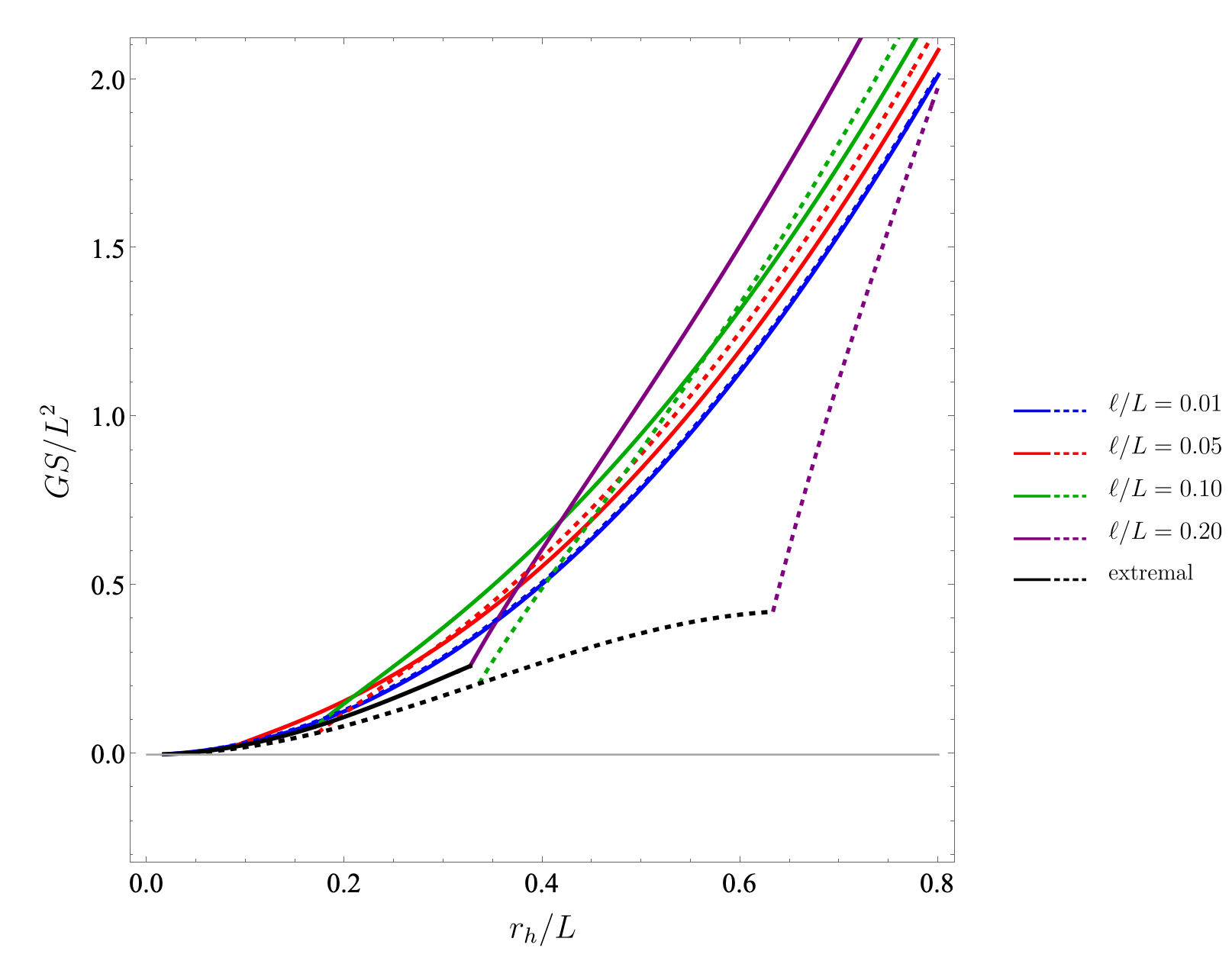}
	\caption{\textbf{Entropy: Models I and II.} Entropy as a function of the horizon radius. The colored solid curves correspond to Model I, whose regular core is dS-like, while the colored dashed curves correspond to Model II, whose regular core is AdS-like. Curves with the same color have the same value of the regularization scale \(\ell/L\). The black solid and dashed curves show the corresponding extremal radii for Models I and II, respectively.}
	\label{fig:entropy}
\end{figure}

The behaviour of the entropy is illustrated in Fig.~\ref{fig:entropy}. The colored solid curves correspond to Model I, while the colored dashed curves correspond to Model II. The black curves denote the corresponding extremal horizon radii. The physical positive-temperature branches start at these extremal curves. Although a formal continuation of the integrated entropy outside the physical branch can develop divergences, these regions lie below the extremal radius and are not part of the thermodynamic branch used in the Hawking--Page analysis. With the normalization adopted here, the entropy on the physical branches shown in the figure is non-negative and approaches the Schwarzschild--AdS area law at large horizon radius.

This prescription should be understood as part of the reduced thermodynamic scheme. In any specified covariant completion, the entropy entering the first law is expected to agree with the corresponding Iyer--Wald Noether charge, up to the standard ambiguities associated with boundary terms and total derivatives. Since the reduced polymerized model does not by itself select a unique off-shell covariant action, we fix the entropy within the reduced description by the requirements stated above: dimensionless logarithms, consistency with the extended Smarr relation, and the correct Schwarzschild--AdS free-energy normalization used in the Hawking--Page analysis.

\section{Phase structure}\label{sec:PT}

Asymptotically AdS black holes provide a particularly natural setting for black hole thermodynamics. The AdS geometry acts effectively as a confining box: freely propagating massive and massless particles do not escape to infinity as they do in asymptotically flat spacetimes. With reflecting boundary conditions imposed at the timelike conformal boundary, Hawking radiation can remain confined and equilibrate with the black hole. This makes the thermodynamic description especially clean, since for sufficiently large black holes, with horizon radius large compared with the AdS length scale, one may assume that the outgoing radiation has had enough time to reach equilibrium with the black hole. The importance of asymptotically AdS spacetimes extends beyond thermodynamics. Their boundary structure also provides the natural arena for holography, allowing bulk gravitational dynamics to be described in terms of a quantum field theory living on the boundary. This is the basis of the AdS/CFT correspondence, which has become one of the most powerful tools for both practical calculations and conceptual progress in quantum gravity. The canonical example is the duality between type IIB string theory in the AdS bulk and $\mathrm{N}=4$ super-Yang–Mills theory on the boundary \cite{M:99}. For this reason, the study of gravitational physics in asymptotically AdS spacetimes is directly connected to the study of strongly coupled gauge theories \cite{SS:02}.

In this asymptotically AdS setting, black holes exhibit the well-known Hawking--Page transition \cite{HP:83}. At low temperatures, small black holes are thermodynamically disfavored and eventually evaporate, leaving thermal radiation in AdS. By contrast, sufficiently large black holes can remain stable: their Hawking radiation is reflected by the AdS boundary and returns in finite time, allowing the black hole to equilibrate with its own radiation. These two phases are separated by a critical temperature $T_{\rm HP}$, at which the Hawking–Page transition occurs.

This behavior differs sharply from the asymptotically flat case. In AdS, the black hole temperature is not a monotonically decreasing function of the horizon size, and large black holes can have positive heat capacity. Through the AdS/CFT correspondence, the Hawking–Page transition is interpreted as a confinement/deconfinement transition in the dual boundary theory. It has therefore played an important role in the study of nonperturbative dynamics in strongly coupled CFTs, as well as in discussions of black hole thermodynamics and the information problem \cite{W:98,H:16}. The equilibrium state of the thermodynamic system corresponds to the global minimization of the free energy. To properly study the phase structure, we need to study the free energy as a function of temperature by drawing parametric plots using the outer horizon radius as a parameter. 

In the extended thermodynamic phase space, the bare cosmological constant is promoted to a pressure $P$, with conjugate thermodynamic volume \(V\) \cite{KRT:09}. With this identification, the black hole mass \(M\) is interpreted as enthalpy rather than internal energy \cite{KMT:17}. In other words, \(M\) includes both the internal energy of the black hole and the energy required to displace the vacuum energy of the surrounding spacetime. Equivalently,
\begin{align}
M=E+PV,
\end{align}
where \(E\) denotes the internal energy.

Allowing $P$, and hence $L$, to vary is not merely a formal device. It can be motivated from a consistent variational principle \cite{UTS:09}, and may also acquire a dynamical interpretation through a three-form gauge potential, as in the Brown--Teitelboim mechanism \cite{BT:87}. We also allow the regularization scale \(\ell\) to vary. This is motivated by the possibility that \(\ell\) is not a fixed external constant, but instead emerges from more fundamental degrees of freedom, for instance through vacuum expectation values of elementary fields \cite{GKK:96,CM:95}, or through backreaction effects during evaporation \cite{CRFLV:23,BBCRG:21,BBCRG:22}. Including variations of both \(P\) and \(\ell\) is also necessary for obtaining the correct Smarr relation \cite{KRT:09} , which follows from the scaling properties of the thermodynamic variables.

For the phase diagrams below, however, we work at fixed $P$ and fixed $\ell$. Each free-energy curve is therefore a canonical ensemble labelled by these two parameters. The relevant thermodynamic potential is
\begin{align}
	F=M-TS .
\end{align}

This quantity should be understood as the Gibbs free energy measured relative to the thermal AdS vacuum of the same branch of the theory. More precisely, for fixed bare couplings the reference background is the $m= 0$ solution obtained from the same reconstruction function $\tilde{f}$. This background is empty AdS with effective radius $L_{\mathrm{eff}}$, determined by Eq.~\eqref{eq:GL-eff}, while the thermodynamic pressure is still associated with the bare cosmological constant. Following the standard convention in higher-curvature, we set the free energy of this thermal AdS saddle to zero and measure the black hole free energy relative to it. With this normalization, the Hawking–Page transition occurs when $F$ changes sign.

Equilibrium configurations at fixed \(T\), \(P\), and \(\ell\) are determined by the global minimum of \(F\). Varying \(P\) then corresponds to comparing a family of such canonical ensembles, rather than varying the pressure along a single thermodynamic curve. Thus, for each fixed pair \((P,\ell)\), the phase structure is obtained by plotting \(F\) parametrically as a function of \(T\), with the horizon radius \(r_h\) used as the parameter.

\begin{figure}[t]
	\centering
	\includegraphics[width=\columnwidth]{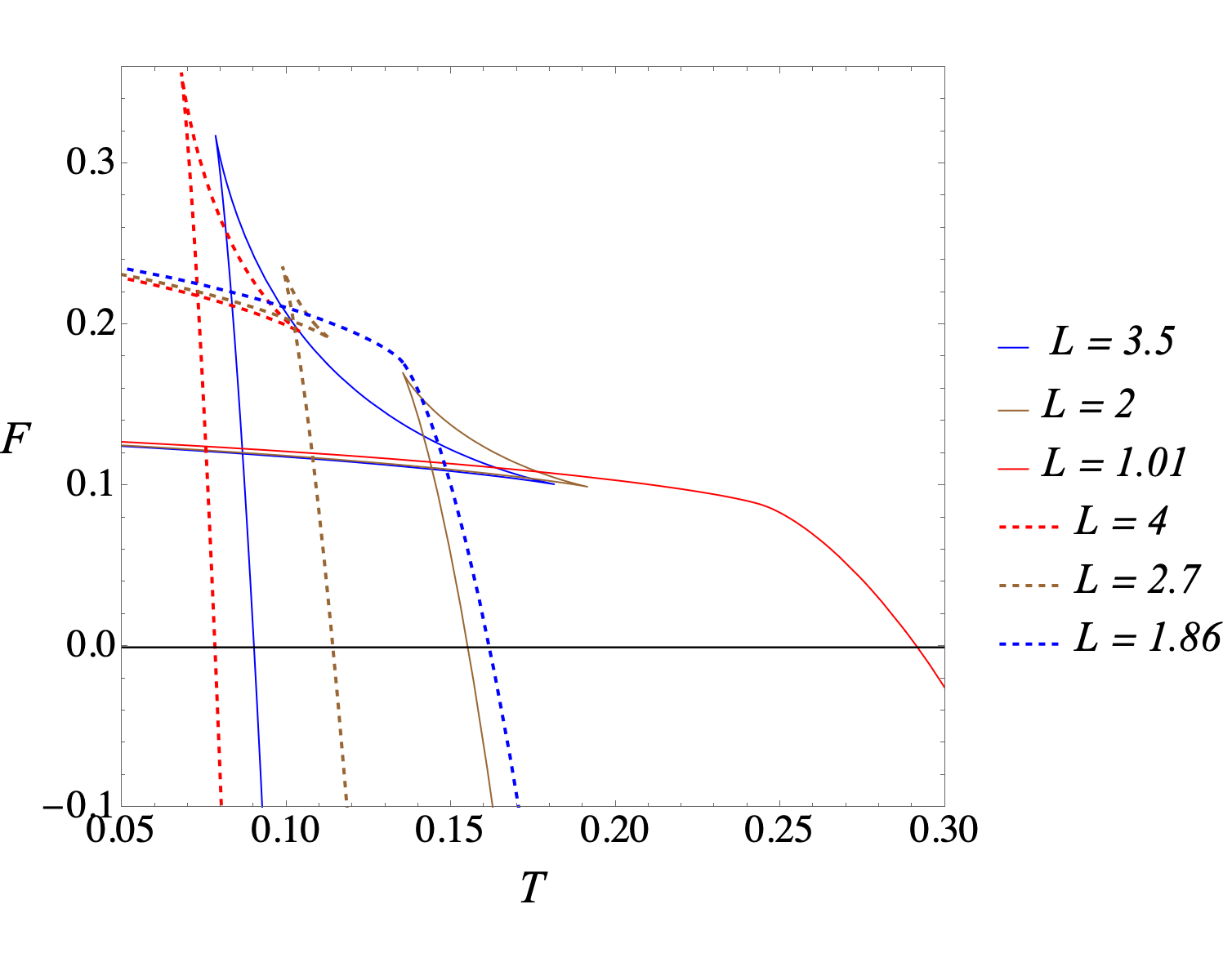}
\caption{\textbf{Free energy: Models I and II.} The free energy \(F\) is plotted parametrically as a function of the temperature \(T\), with the horizon radius used as the parameter. Solid curves correspond to Model I, which has a dS core, while dashed curves correspond to Model II, which has an AdS core. The minimal-length scale is fixed to \(\ell=0.1\), and the cosmological-constant parameter \(L\) is varied. We work in units with \(G=1\).}
	\label{fig:PT}
\end{figure}

The representative free energy diagrams for Models I and II are shown in Fig.~\ref{fig:PT}. Although the explicit thermodynamic functions are model dependent, the qualitative branch structure is the same for the model pairs considered in this work. The thermodynamically preferred configuration is selected by the crossing of the black hole free energy with the thermal AdS background. The dominant transition is therefore of Hawking–Page type. A small/large black hole structure can still appear at the level of subdominant saddles, but in the present models the corresponding critical behavior occurs at positive free energy and does not define the dominant first-order transition of the canonical ensemble.

This behaviour is consistent with the phase structure found for vacuum RBHs in quasitopological gravity in $D = 5$, where core regularization alone does not produce a physical small/large black hole phase transition. In higher dimensions, such as $D = 7$ or $D = 9$, a small/large transition can become thermodynamically relevant if it occurs at negative free energy, below the thermal AdS background \cite{HKMS:25}. The situation is also different for RBHs supported by nonlinear electrodynamics, where a magnetic charge introduces an additional thermodynamic scale. At fixed charge, complete evaporation to thermal AdS is not available in the same way, and the small/large black hole transition can become the dominant thermodynamic feature \cite{SS:24,S:24}. In contrast, within the vacuum RBHs considered here, the Hawking–Page transition is the robust feature of the canonical ensemble.

The regular core nevertheless leaves a quantitative imprint on the transition temperature. The Hawking–Page temperature for Models I and II is shown in Fig.~\ref{fig:Thp}. The dS core and AdS core branches have the same qualitative Hawking–Page structure, but the crossing point $F=0$ occurs at different temperatures. This ordering should not be attributed simply to the sign of the limiting core curvature. The Hawking–Page transition is determined by the full free-energy balance along the physical outer-horizon branch, including the mass function, temperature, entropy correction, and the endpoint structure selected by the reconstruction function. In the large AdS radius part of the branch, the dS core solution has a higher Hawking–Page temperature than its AdS core counterpart, while the ordering can be modified close to the lower admissible range of the AdS scale. Thus the asymptotically AdS structure controls the existence and type of transition, whereas the regularization core controls its quantitative location.

\begin{figure*}[t]
	\centering
	\includegraphics[width=\textwidth]{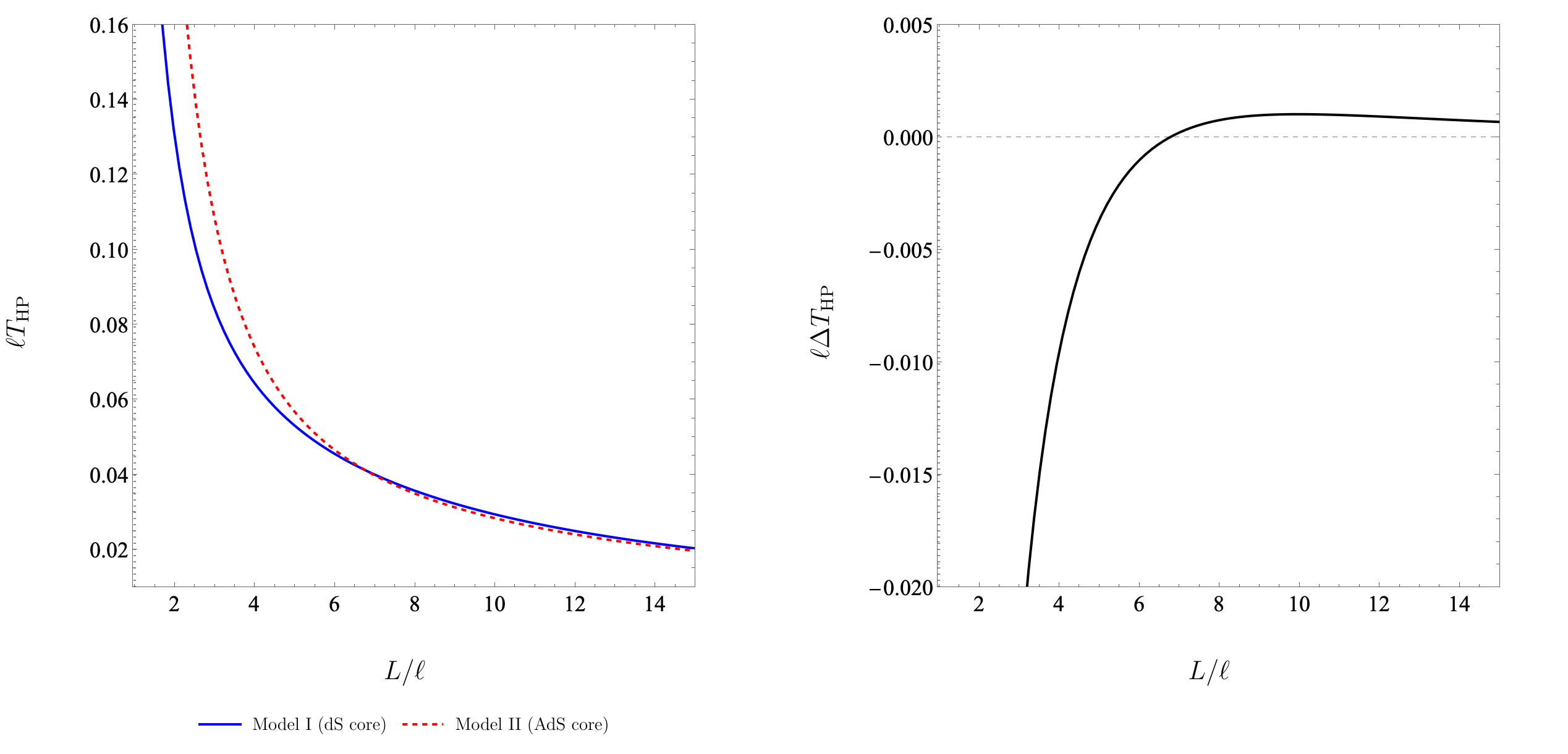}
	\caption{\textbf{Hawking--Page temperature for Models I and II.}
Left panel: dimensionless Hawking--Page temperature $\ell T_{\rm HP}$ as a function of the dimensionless AdS scale $L/\ell$ for Model I, which has a dS core, and Model II, which has an AdS core. Right panel: corresponding dimensionless difference
$\ell \Delta T_{\rm HP}=\ell\left(T^{\rm I}_{\rm HP}-T^{\rm II}_{\rm HP}\right)$.
The vertical dotted line marks the lower admissible boundary $L/\ell=1$, required for the correct asymptotic AdS behaviour as shown in Table~\ref{tab:polymerization-models}. The ordering of the transition temperatures reverses close to this boundary, while for sufficiently large $L/\ell$ one finds $T^{\rm I}_{\rm HP}>T^{\rm II}_{\rm HP}$.}	\label{fig:Thp}
\end{figure*}

\subsection{Equation of state}\label{sec:EoS}

We now turn to the equation of state. The purpose of this analysis is diagnostic: the \(P-v\) plane shows how the regularization scale modifies the thermodynamic response of the black hole branch. In particular, the regularized models can display a finite-volume pressure divergence reminiscent of the excluded-volume singularity of a van der Waals fluid. This analogy should be understood in a thermodynamic sense only. It does not imply that the black hole undergoes a physical van der Waals small/large transition, nor that the horizon degrees of freedom behave as ordinary molecules with a literal hard core. The physical transition identified above remains the Hawking–Page transition, obtained by comparing the black hole free energy with thermal AdS.

The useful point of comparison is the ordinary van der Waals equation,
\begin{align}
	P=\frac{T}{v-b}-\frac{a}{v^2},
\end{align}
where $a$ measures attractive interactions and $b$ is an excluded-volume parameter. In black hole thermodynamics the specific volume is an effective thermodynamic measure of the horizon scale rather than a microscopic molecular volume. In four dimensions, and in the convention $G= 1$, it is
\begin{align}
	v=2r_h.
\end{align}
Using the expressions for the mass, temperature, and pressure, the equation of state for the present class of models can be written in terms of $T$ and $r_h$ as
\begin{widetext}
	\begin{align}
		P=\frac{T}{2Gr_h\tilde{f}'(\tilde{f}^{-1}(1/r^2_h))}+\left(\frac{1}{4\pi G r^2_h \tilde{f}'(\tilde{f}^{-1}(1/r^2_h))}-\frac{3}{8\pi G}\tilde{f}^{-1}(1/r^2_h)\right).\label{eq:EoS-gen}
	\end{align}
\end{widetext}
In the GR limit,
\begin{align}
	\tilde{f}(u)=u, \quad \tilde{f}'(u)=1,
\end{align}
this reduces to the Schwarzschild–AdS equation of state
\begin{align}
	P=\frac{T}{2G r_h}-\frac{1}{8\pi G r_h^2},
\end{align}
For $G=1$ and $v=2r_h$, this becomes 
\begin{align}
	P=\frac{T}{v}-\frac{1}{2\pi v^2}.
\end{align}
The first term has the ideal gas form, while the second is the attractive gravitational correction. Away from the GR limit, the reconstruction function modifies both the coefficient of the thermal term and the temperature-independent contribution.
Let $u_h$ denote the value of the curvature variable on a real physical horizon branch,
\begin{align}
	u_h=\tilde f^{-1}(1/r_h^2),
	\qquad
	\tilde f(u_h)=\frac{1}{r_h^2}.
\end{align}
On any regular segment for which $\tilde{f}'(u_h)\neq 0$, the temperature-independent part of Eq.~\eqref{eq:EoS-gen} is
\begin{align}
	\mathcal{A}(r_h)=\frac{1}{4\pi G r^2_h\tilde{f}'(u_h)}-\frac{3}{8\pi G}u_h,
\end{align}
or equivalently 
\begin{align}
	\mathcal A(r_h)
	=
	\frac{\tilde f(u_h)}{8\pi G\,\tilde f'(u_h)}
	\left[
	2-3\nu(u_h)
	\right].
\end{align}
It is useful to introduce the algebraic growth index 
\begin{align}
	\nu(u_h)=\frac{u_h\tilde{f}'(u_h)}{\tilde{f}(u_h)}.
\end{align}
Then 
\begin{align}
	\mathcal{A}(r_h)=\frac{\tilde{f}(u_h)}{8\pi G \tilde{f}'(u_h)}\left(2-3\nu(u_h)\right).
\end{align}
Since $\tilde{f}(u_h)=1/r^2_h>0$, the sign of this term is controlled by the factor \(2-3\nu(u_h)\) and $\tilde{f}'(u_h)$. In the Schwarzschild--AdS limit, \(\tilde f(u)=u\), so \(\nu=1\), and the temperature-independent term reduces to the usual negative gravitational correction to the ideal gas term. In this equation of state sense, a negative \(\mathcal A\) is attractive, while a positive \(\mathcal A\) is effectively repulsive.

The finite volume divergence visible in Fig.~\ref{fig:Pv} is not a restatement of regularity at $r = 0$. It is a branch-endpoint effect in the horizon equation. The pressure contains the factor $1/\tilde{f}'(u_h)$, and the horizon condition fixes $u_h$ implicitly through $\tilde{f}(u_h)=1/r^2_h$. Therefore a pressure divergence at finite specific volume occurs when the physical inverse branch degenerates at a finite endpoint with
\begin{align}
	\tilde{f}'(u_h)\to 0.
\end{align}
This is the precise sense in which the regularization scale produces an excluded-volume-like behaviour. The lower end-point restricts the accessible horizon volumes, but it does not imply a literal microscopic hard core interpretation.

For Model I,
\begin{align}
	u_h=\frac{1}{r_h^2-\ell^2},
	\qquad
	\tilde f'(u_h)=\frac{(r_h^2-\ell^2)^2}{r_h^4}.
\end{align}
The physical branch exists for $r_h> \ell$. Substitution into Eq.~\eqref{eq:EoS-gen} gives
\begin{align}
	P_{\mathrm{I}}=
	\frac{T r_h^3}{2G(r_h^2-\ell^2)^2}
	+
	\frac{3\ell^2-r_h^2}{8\pi G(r_h^2-\ell^2)^2}.
\end{align}
Thus $P_{\rm I}$ diverges as $r_h\to \ell^{+}$, or equivalently as $v\to 2\ell$ when $G=1$ (See Fig.~\ref{fig:Pv}). The sign-change criterion is also explicit. For this model,
\begin{align}
	\nu(u_h)=\frac{1}{1+\ell^2u_h}.
\end{align}
The condition $\nu_{\rm I} = 2/3$ gives
\begin{align}
	u_{h,*}=\frac{1}{2\ell^2},
\end{align}
and hence
\begin{align}
	r^{\rm I}_{h,*}=\sqrt{3}\ell, \qquad v^{\rm I}_{*}=2\sqrt{3}\ell.
\end{align}
The temperature-independent contribution is therefore attractive for $r_h>\sqrt{3}\ell$ and repulsive for $\ell<r_h<\sqrt{3}\ell$, close to the finite-volume endpoint.

For Model II, the endpoint of the horizon branch should be distinguished from
the sign change responsible for the AdS core. For the representative choice
\(\epsilon=1\),
\begin{align}
	\tilde{f}_{\rm II}(u)
	=
	\frac{u(1-\ell^2 u)}{(1+\ell^2 u)^2}.
\end{align}
Introducing \(y=\ell^2 u\), this becomes
\begin{align}
	\tilde{f}_{\rm II}(u)
	=
	\frac{1}{\ell^2}
	\frac{y(1-y)}{(1+y)^2}.
\end{align}
The function is positive for \(0<y<1\) and changes sign at \(y=1\). The
AdS-core region therefore lies beyond the zero of \(\tilde f_{\rm II}\), and
is not itself part of the positive horizon branch determined by
\(\tilde f_{\rm II}(u_h)=1/r_h^2>0\).

The branch continuously connected to the Schwarzschild--AdS regime is the
monotonic segment \(0<y<1/3\). Indeed,
\begin{align}
	\frac{d}{dy}
	\left[
	\frac{y(1-y)}{(1+y)^2}
	\right]
	=
	\frac{1-3y}{(1+y)^3},
\end{align}
so this branch ends at the maximum \(y_{\rm end}=1/3\), where
\(\tilde f_{\rm II}=1/(8\ell^2)\). The horizon condition then gives
\begin{align}
	r_{h,\min}^{\rm II}=2\sqrt{2}\ell,
	\qquad
	v_{\min}^{\rm II}=4\sqrt{2}\ell .
\end{align}

For this branch,
\begin{align}
	\nu_{\rm II}(y)
	=
	\frac{y\,\tilde f'_{\rm II}(u)}{\tilde f_{\rm II}(u)}
	=
	\frac{1-3y}{1-y^2}.
\end{align}
Equivalently, the temperature-independent part of the equation of state can be
written as
\begin{align}
	\mathcal A_{\rm II}(y)
	=
	\frac{y(-1+9y-2y^2)}
	{8\pi G\ell^2(1-3y)} ,
	\qquad
	0<y<\frac13 .
\end{align}
Thus \(\mathcal A_{\rm II}\) changes sign when
\begin{align}
	-1+9y-2y^2=0,
\end{align}
or
\begin{align}
	y_\ast=\frac{9-\sqrt{73}}{4}.
\end{align}
Using
\begin{align}
	\frac{1}{r_h^2}
	=
	\frac{1}{\ell^2}
	\frac{y(1-y)}{(1+y)^2},
\end{align}
this corresponds to
\begin{align}
	r_{h,\ast}^{\rm II}
	=
	\ell
	\frac{1+y_\ast}{\sqrt{y_\ast(1-y_\ast)}}
	\simeq3.505\,\ell ,
\end{align}
or \(v_\ast^{\rm II}\simeq7.010\,\ell\). The divergence at
\(r_{h,\min}^{\rm II}=2\sqrt{2}\ell\) instead comes from
\(\tilde f'_{\rm II}(u_h)\to0\). Thus the pressure divergence is controlled by
the endpoint of the positive GR-connected branch, not directly by the
AdS-core sign change of \(\tilde f_{\rm II}\).

The distinction between dS core and AdS core branches can now be stated in a branch-independent way. For dS core geometries the reconstruction function typically approaches a positive limiting value,
\begin{align}
	\tilde f(u)\to F_0>0 ,
	\qquad u\to+\infty .
\end{align}
When this saturation occurs on the physical positive branch, the horizon equation implies a finite lower endpoint,
\begin{align}
	r_{h,\min}=\frac{1}{\sqrt{F_0}} .
\end{align}
This endpoint should not be confused with the extremal black hole radius, which is determined by the simultaneous conditions \(f(r_h)=0\) and \(T=0\). The endpoint \(r_{h,\min}\) is instead the lower limit of the horizon branch appearing in the equation of state; the positive-temperature thermodynamic branch may begin at a larger radius.

For AdS-core geometries the limiting core value is negative,
\begin{align}
	\tilde f(u)\to F_0<0 ,
	\qquad u\to+\infty ,
\end{align}
whereas the horizon equation requires \(\tilde f(u)>0\). The deep core therefore does not itself lie on the positive horizon branch. A finite-volume pressure divergence can occur only if the intermediate positive branch of \(\tilde f\) has a finite maximum. If
\begin{align}
	\psi_{\max}=\max_{\mathcal I_+}\tilde f(u),
\end{align}
where \(\mathcal I_+\) denotes the positive horizon branch, then the corresponding lower bound is
\begin{align}
	r_h\ge \frac{1}{\sqrt{\psi_{\max}}}.
\end{align}
Thus, for AdS core models, the endpoint is controlled by the global structure of the positive horizon branch, not by the limiting negative core curvature itself.

This also clarifies the relation with quasitopological constructions. There, regularity may be implemented through a divergence of the theory function \(h_f\), so that \(h_f^{-1}\) saturates at a finite curvature value. In the present polymerized reconstruction, the metric is written directly in terms of \(\tilde f\). Regularity therefore requires boundedness and smooth saturation of \(\tilde f\), while a finite-volume pressure divergence follows only when the physical positive branch has a finite endpoint at which the inverse branch degenerates. This endpoint-controlled behaviour is one of the main thermodynamic differences between the dS core and AdS core regularizations considered here.

\begin{figure}[t]
	\centering
	\includegraphics[width=\columnwidth]{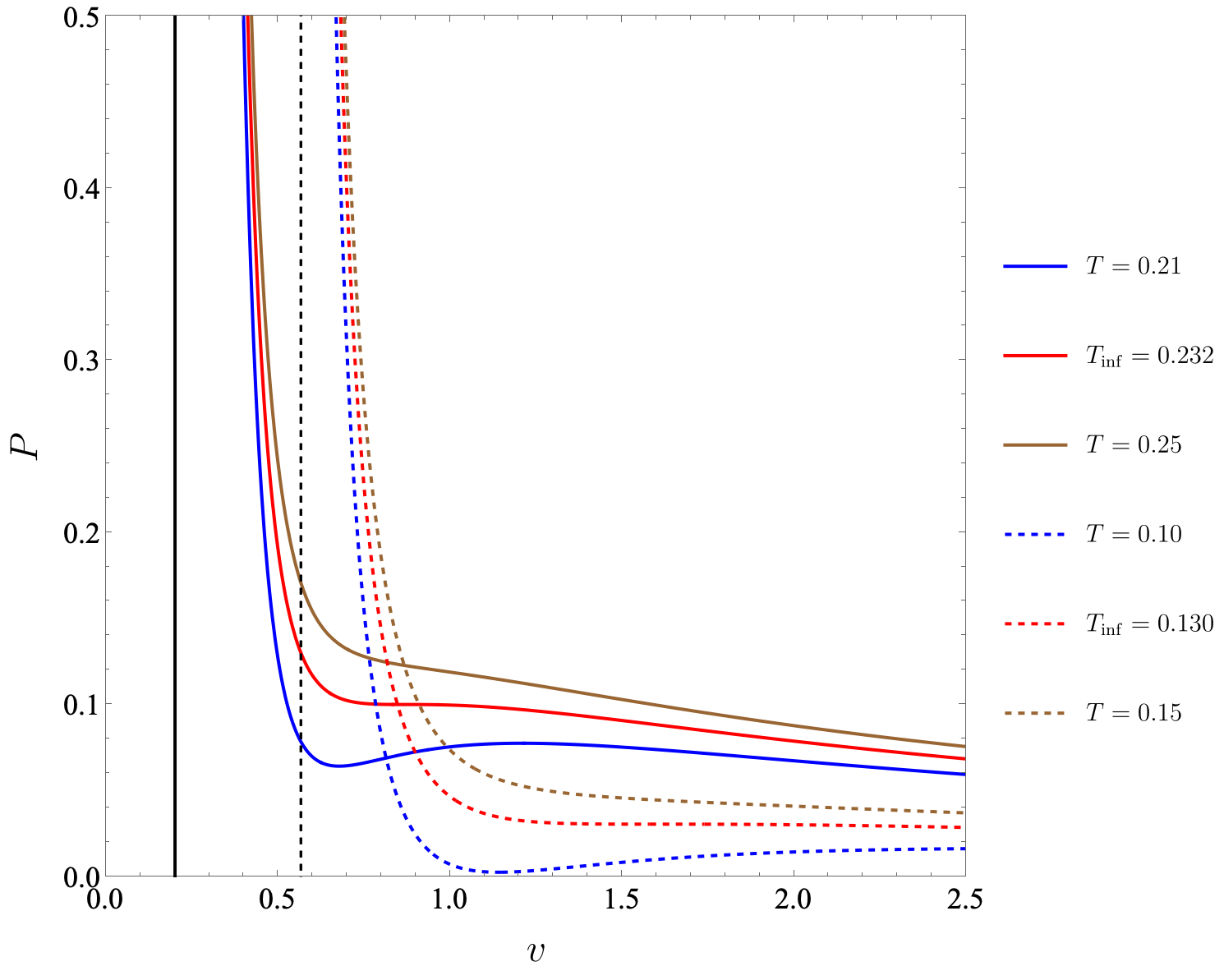}
	\caption{\textbf{$P-v$ diagram: Models I and II.} The isotherms for Model I, which has a dS core, are shown with solid lines, while those for Model II, which has an AdS core, are shown with dashed lines. We fix the regularization scale to $\ell=0.1$ and work in units $G=1$. For each model we display three representative isotherms, corresponding to temperatures below, at, and above the inflection temperature \(T_{\mathrm{inf}}\). The vertical black lines mark the branch endpoints at which the pressure diverges. The solid vertical line corresponds to Model I and is located at \(v_{\min}^{\rm I}=2\ell=0.2\), while the dashed vertical line corresponds to Model II and is located at \(v_{\min}^{\rm II}=4\sqrt{2}\ell=0.4\sqrt{2}\).}
	\label{fig:Pv}
\end{figure}

\section{Conclusions}\label{sec:conclusions}

In this work we studied the thermodynamics of a class of polymerized vacuum RBHs in asymptotically AdS spacetime. The construction is motivated by effective loop quantum gravity dynamics and is formulated in a deparameterized LTB framework, where the dust field is used only as an internal clock. The dust does not act as a physical matter source on the vacuum branch. This distinction is important: the RBHs considered here are not supported by an additional matter sector, but arise from a modified vacuum dynamics encoded in the polymerization function.

We first extended the shell-based polymerized vacuum construction to include a negative cosmological constant. In the vacuum sector the shell Hamiltonian remains conserved, and the Birkhoff-type condition ensures that the static geometry is uniquely reconstructed from the mass parameter once the polymerization function and the bare AdS scale are specified. This led to the universal metric form Eq.~\eqref{eq:f-LTB} which provides the starting point for the thermodynamic analysis. We also clarified the conditions under which the center is curvature-regular. In this formulation, regularity is controlled by the boundedness and smoothness properties of the reconstructed function $\tilde f(u(r))$, rather than by the introduction of an explicit matter profile.

A central feature of the framework is that it naturally accommodates both dS core and AdS core RBHs. This allowed us to compare two distinct vacuum resolutions of the classical singularity within the same thermodynamic setting. We constructed representative pairs of models, in which the AdS core geometries are obtained as deformations of the corresponding dS core solutions, and determined their effective asymptotic parameters. Although both types of cores remove the curvature singularity, they do not lead to identical thermodynamic behavior.

We then developed the extended thermodynamics of these solutions. The bare cosmological constant was interpreted as a thermodynamic pressure, while the black hole mass was treated as enthalpy. The entropy was obtained by integrating the first law and fixing the integration constant so that logarithmic terms have dimensionless arguments, the extended Smarr relation is satisfied, and the Schwarzschild–AdS free-energy normalization is recovered in the appropriate limit. This provides a consistent reduced thermodynamic scheme for the polymerized vacuum geometries considered here.

The main result of our phase structure analysis is that the dominant transition is of Hawking–Page type. The black hole free energy is compared with the corresponding thermal AdS background, and the transition occurs when the black hole branch becomes thermodynamically preferred. Although a small/large black hole branch structure may appear at the level of subdominant saddles, it does not define the physical first-order transition of the canonical ensemble for the models studied here. In this sense, the Hawking–Page transition is the robust thermodynamic feature of these polymerized vacuum RBHs in AdS.

The comparison between dS core and AdS core solutions reveals a robust quantitative distinction, but it should be stated with an important qualification. Along the physical outer-horizon branch and in the large-\(L/\ell\) regime, the dS core members of the model pairs considered here undergo the Hawking--Page transition at a higher temperature than their AdS core counterparts. Close to the lower admissible range of \(L/\ell\), this ordering can be modified, as shown by the numerical results. Thus the sign of the limiting core curvature alone is not the direct explanation of the transition temperature. The Hawking--Page point is determined by the full free-energy balance along the physical branch: the reconstruction function changes the mass-radius relation, the surface-gravity temperature, the entropy correction, and the endpoint structure of the allowed horizons. The regular core therefore affects the thermodynamics by deforming the physical horizon branch relative to the asymptotically AdS background. In this sense, the asymptotic AdS structure controls the existence and type of the Hawking--Page transition, while the regularization core controls its quantitative location.

We also analyzed the equation of state in the $P-v$ plane. The resulting diagrams display a finite-volume divergence reminiscent of the excluded-volume behavior of a van der Waals fluid. However, this divergence should not be interpreted as evidence for a physical van der Waals small/large black hole transition. Instead, it originates from the endpoint of the physical horizon branch selected by the reconstruction function. More precisely, the pressure divergence is controlled by the degeneration of the inverse branch of the horizon equation, rather than by regularity alone or by the local core geometry at $r=0$. This distinction is important because it separates the existence of a regular center from the thermodynamic branch structure of the black hole. 

Our results therefore show that polymerized vacuum RBHs in AdS provide a useful arena for studying how singularity resolution can affect black hole thermodynamics without introducing additional matter charges. The AdS boundary fixes the canonical ensemble and makes the Hawking–Page transition meaningful, while the polymerization function controls both the regular core and the endpoint structure of the physical horizon branch. Within the class of models studied here, the qualitative phase structure is robust, but the transition temperature is sensitive to whether the singularity is replaced by a dS or AdS core.

Several extensions would be worthwhile. First, it would be important to analyze explicit fully covariant completions, such as the construction proposed in Ref.~\cite{LS:26}. Although the thermodynamic results obtained here are self-consistent within the reduced effective prescription, determining the corresponding covariant conserved mass and Iyer--Wald entropy, and verifying their relation to the reduced quantities used in this work, requires specifying a particular covariant action together with its boundary terms. Applying the covariant phase-space and Iyer--Wald formalisms to such a completion would therefore clarify the thermodynamic interpretation and provide a direct way to assess the consequences of the non-uniqueness of the covariant completion.

A second direction would be to formulate the thermodynamics using the Euclidean path-integral approach. This would be especially useful for extending the analysis to asymptotically dS spacetimes, where one could study the canonical ensemble by placing the black hole inside an isothermal cavity. Such an analysis may clarify whether the endpoint structure found here persists in the corresponding dS thermodynamic setting.

Another interesting question concerns the definition of temperature in covariant completions involving an additional scalar field. The model of Ref.~\cite{LS:26}, for example, contains a scalar sector which may be interpreted in terms of a shift-symmetric Lagrangian. It would therefore be interesting to examine whether the notion of temperature is modified in this setting, along the lines of Ref.~\cite{LHK:23}, and to compare the result with other possible covariant completions.

Finally, it would be important to study genuinely dynamical processes, such as black hole evaporation or collapse, in order to understand how the thermodynamic endpoint structure is connected to the global causal structure of the corresponding regular geometries. In the spirit of Refs.~\cite{MS:23-thermo,DSST:23,T:26,MS:23,MMT:22,APS:26}, one could investigate whether the thermodynamic endpoint identified here has a counterpart in the late-time evolution or near core dynamics of RBH spacetimes.

Overall, the analysis demonstrates that the thermodynamics of polymerized vacuum black holes is not determined by regularity alone. The same requirement of singularity resolution can lead to different core geometries, different effective asymptotic parameters, and different Hawking–Page temperatures. The framework developed here therefore provides a controlled setting in which the thermodynamic consequences of quantum gravity-inspired vacuum regularization can be systematically investigated.

\section*{acknowledgments}
We would like to thank Hongguang Liu for useful discussions and helpful comments. S.B. is supported by the National Natural Science Foundation of China under grant Nos. 12275238, W2433018, the National Key
Research and Development Program under grant No. 2020YFC2201503, and the Zhejiang Provincial Natural
Science Foundation of China under grant Nos. LR21A050001 and LY20A050002, and the Fundamental Research Funds for the Provincial Universities of Zhejiang in China under grant No. RF-A2019015. I.S. is supported by the Institute for Theoretical Sciences at Westlake University.

\appendix

\section{Thermodynamic quantities} \label{sec:app:thermo}

In this appendix, we present the thermodynamic quantities used to construct the phase diagrams. For models admitting compact analytic expressions, we display the relevant formulas explicitly. For the remaining models, the expressions are either too cumbersome or not available in closed analytic form, and are therefore omitted from the text. Their thermodynamic properties are analyzed numerically instead. These numerical checks confirm that the qualitative features discussed below are not specific to Models I and II, which are used as representative examples throughout the paper. Throughout this appendix, we set $G=1$.

For Model I, the mass function and Hawking temperature are given, respectively, by
\begin{align}
	m_{\mathrm I}(r_h)
	=
	\frac{r_h^3\left(-\ell^2+L^2+r_h^2\right)}
	{2 L^2\left(-\ell^2+r_h^2\right)} ,
\end{align}
and
\begin{align}
	T_{\mathrm I}(r_h)
	=
	\frac{
		3\ell^4
		-3\ell^2\left(L^2+2r_h^2\right)
		+r_h^2\left(L^2+3r_h^2\right)
	}
	{4\pi L^2 r_h^3}.
\end{align}
The entropy is obtained by integrating the first law. This gives
\begin{align}
	S_{\mathrm I}(r_h)
	=
	\pi\left[
	\frac{2\ell^2 r_h^2-r_h^4}{\ell^2-r_h^2}
	+
	2\ell^2\log\left(r_h^2-\ell^2\right)
	\right]
	+c .
\end{align}
The integration constant is fixed by requiring the logarithm to have a dimensionless argument. Since the only additional length scale in the problem is the regularization scale $\ell$, we choose
\begin{align}
	c=-2\pi\ell^2\log\ell^2 .
\end{align}
The entropy can therefore be written as
\begin{align}
	S_{\mathrm I}(r_h)
	=
	\pi\left[
	\frac{2\ell^2 r_h^2-r_h^4}{\ell^2-r_h^2}
	+
	2\ell^2\log\left(
	\frac{r_h^2-\ell^2}{\ell^2}
	\right)
	\right].
\end{align}

The free energy is defined by
\begin{align}
	F=M-TS .
\end{align}
For Model I, this yields
\begin{widetext}
\begin{align}
	F_{\mathrm I}(r_h)
	=
	\frac{1}{4L^2 r_h^3}
	\Bigg[
	\frac{2r_h^6\left(-\ell^2+L^2+r_h^2\right)}
	{-\ell^2+r_h^2}
	&-
	\left(
	3\ell^4
	-3\ell^2\left(L^2+2r_h^2\right)
	+r_h^2\left(L^2+3r_h^2\right)
	\right)
	\left(
	\frac{2\ell^2 r_h^2-r_h^4}{\ell^2-r_h^2}
	+
	2\ell^2\log\left(
	\frac{r_h^2-\ell^2}{\ell^2}
	\right)
	\right)
	\Bigg].
\end{align}
\end{widetext}
This choice of integration constant is also compatible with the normalization of the thermal AdS background. In particular, the free energy of the $m=0$ thermal AdS saddle is taken to vanish. In the present parametrization, this background normalization is recovered by first taking the GR limit $\ell\to0$, and then taking the horizon radius $r_h\to0$.
\begin{figure}[t]
	\centering
	\includegraphics[width=\columnwidth]{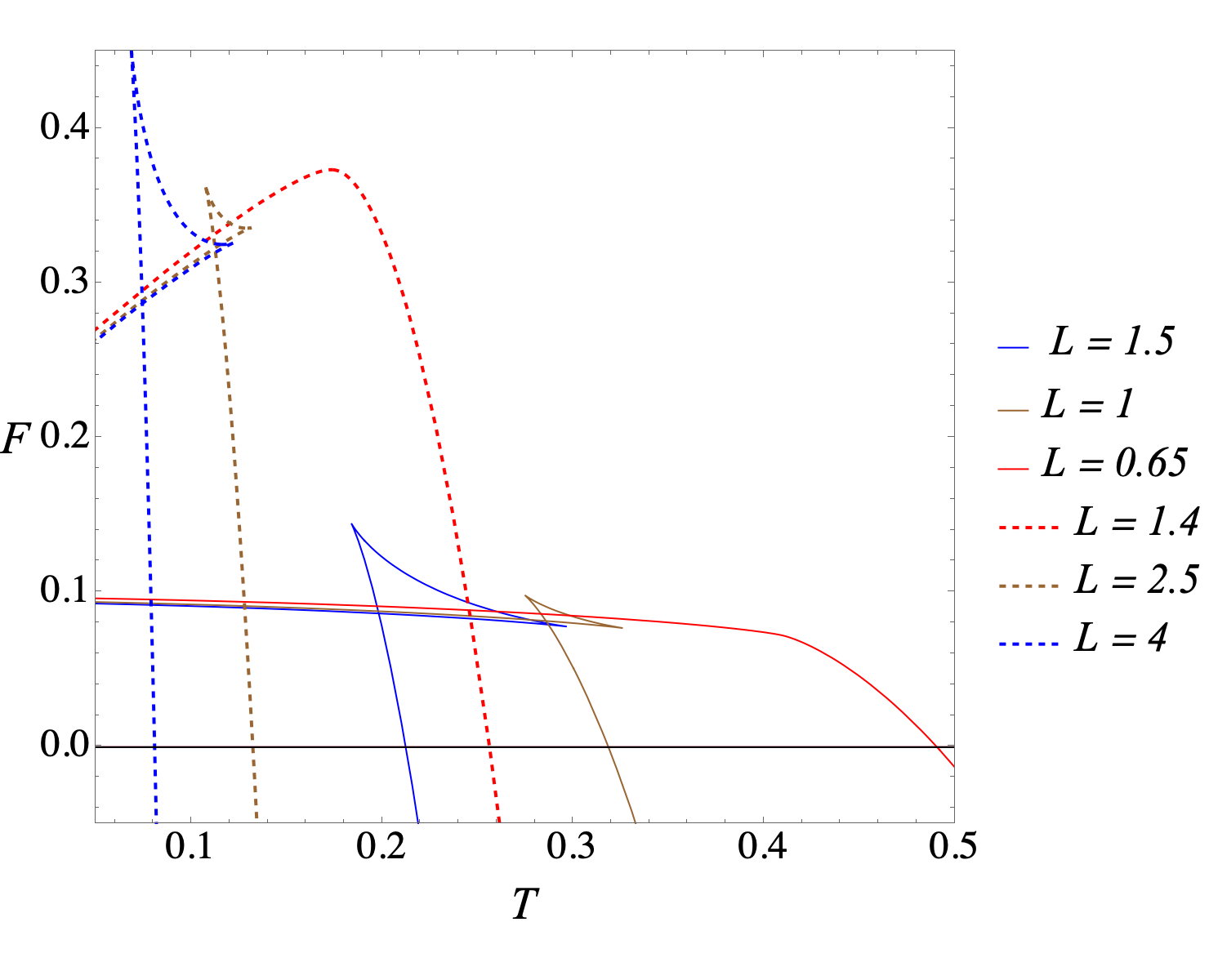}
	\caption{\textbf{Free energy: Models III and IV.} The free energy \(F\) is plotted parametrically as a function of the temperature \(T\), with the horizon radius used as the parameter. Solid curves correspond to Model III, which has a dS core, while dashed curves correspond to Model IV, which has an AdS core. The minimal-length scale is fixed to \(\ell=0.1\), and the cosmological-constant parameter \(L\) is varied. We work in units with \(G=1\).}
	\label{fig:PT-34}
\end{figure}

\begin{figure*}[t]
	\centering
	\includegraphics[width=\textwidth]{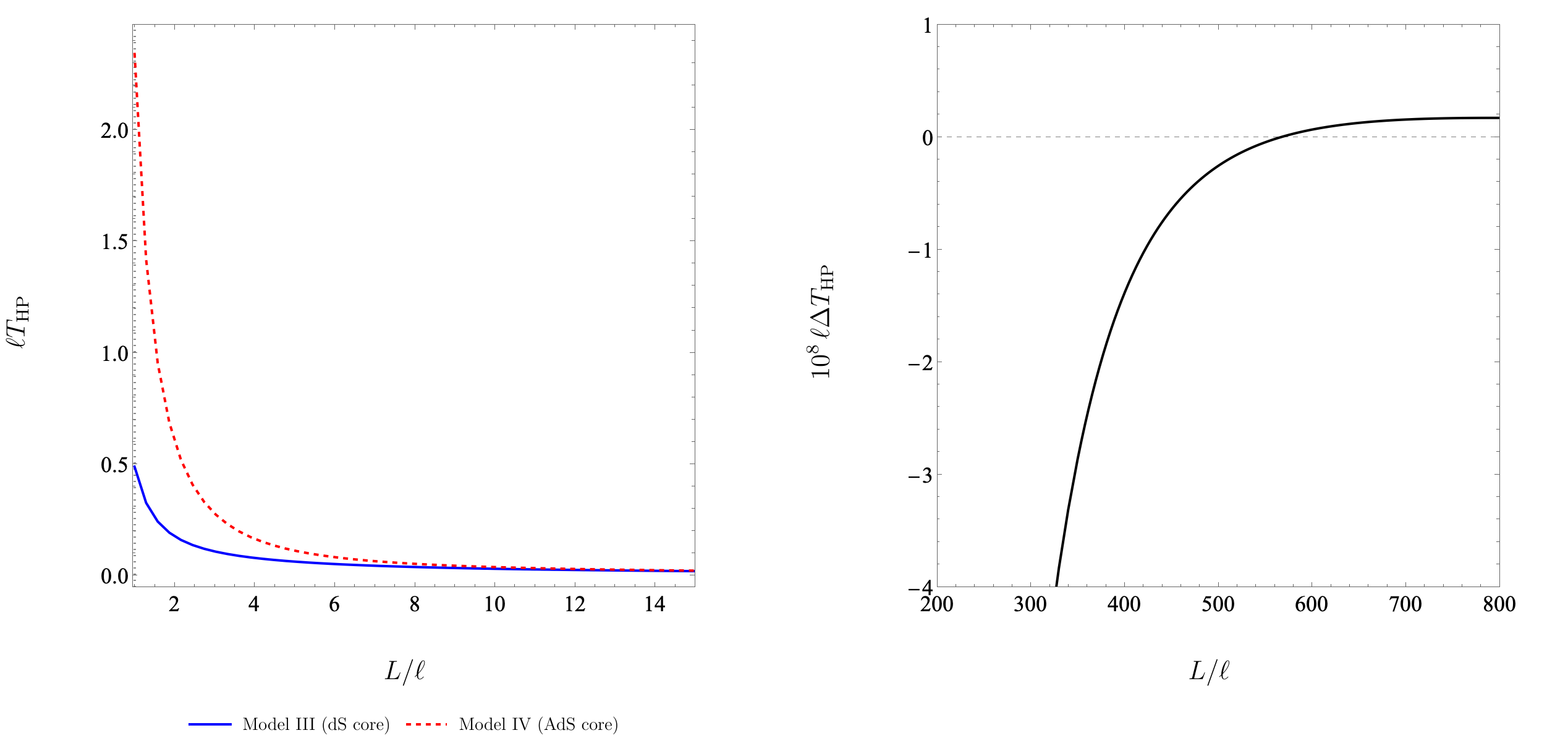}
	\caption{\textbf{Hawking--Page temperature for Models III and IV.}
		Left panel: dimensionless Hawking--Page temperature $\ell T_{\rm HP}$ as a function of the dimensionless AdS scale $L/\ell$ for Model III, which has a dS core, and Model IV, which has an AdS core. Right panel: corresponding dimensionless difference
		$\ell \Delta T_{\rm HP}=\ell\left(T^{\rm {III}}_{\rm HP}-T^{\rm {IV}}_{\rm HP}\right)$.
		The vertical dotted line marks the lower admissible boundary $L/\ell=1$, required for the correct asymptotic AdS behaviour as shown in Table~\ref{tab:polymerization-models}. The ordering of the transition temperatures reverses close to this boundary, while for sufficiently large $L/\ell$ one finds $T^{\rm III}_{\rm HP}>T^{\rm IV}_{\rm HP}$.}	\label{fig:Thp-III-IV}
\end{figure*}

For completeness, we also give the conjugate potentials entering the
Smarr relation. These make it straightforward to verify explicitly that
Eq.~\eqref{eq:smarr} is satisfied. From Eq.~\eqref{eq:conjugates}, the
thermodynamic volume for Model I is
\begin{align}
	V_{\mathrm I}(r_h)=\frac{4\pi r_h^3}{3}.
\end{align}
The potential conjugate to the minimal-length scale is
\begin{widetext}
\begin{align}
	\Psi_{\mathrm I}(r_h)&=
	\frac{\ell}{2L^2 r_h^3\left(r_h^2-\ell^2\right)}
	\Bigg\{
	r_h^2
	\left[
	6\ell^4
	+5L^2 r_h^2
	+9r_h^4
	-3\ell^2\left(2L^2+5r_h^2\right)
	\right]\nonumber\\
	&-2\left(r_h^2-\ell^2\right)
	\left[
	3\ell^4
	-3\ell^2\left(L^2+2r_h^2\right)
	+r_h^2\left(L^2+3r_h^2\right)
	\right]
	\log\left(\frac{r_h^2}{\ell^2}-1\right)
	\Bigg\}.
\end{align}
\end{widetext}
For Model II, the corresponding expressions are considerably more involved. The mass function and Hawking temperature are given, respectively, by
\begin{widetext}
	\begin{align}
		m_{\mathrm{II}}(r_h)
		&=
		\frac{
			-\ell^2 L^4 r_h^3
			+L^4 r_h^5
			+\epsilon \ell^2 L^2 r_h^3\left(2\ell^2-L^2+2r_h^2\right)
			-\Delta_{\epsilon}}
		{4\epsilon \ell^2 L^4\left(\ell^2+r_h^2\right)} ,
		\label{eq:m-model-II}
	\end{align}
	\begin{align}
		T_{\mathrm{II}}(r_h)
		&=
		-\frac{
			2L^{10}r_h^{10}\mathcal{N}_{\mathrm{II}}(r_h)
		}{
			\pi\ell^2
			\big[
			L^4r_h^3((\epsilon-1)\ell^2-3r_h^2)+\Delta_\epsilon
			\big]^2
			\big[
			-L^4r_h^3((\epsilon-1)\ell^2+(1+2\epsilon)r_h^2)+\Delta_\epsilon
			\big]^2
		},
		\label{eq:T-model-II}
	\end{align}
	\begin{align}
		\mathcal{N}_{\mathrm{II}}(r_h)
		&=
		\epsilon^4\ell^6L^4r_h^3
		\big[3\ell^4-L^2r_h^2+3r_h^4-3\ell^2(L^2-2r_h^2)\big]
		\nonumber\\
		&\quad
		-2L^2(3\ell^6-7\ell^4r_h^2+\ell^2r_h^4+3r_h^6)
		\big[L^4r_h^3(\ell^2-r_h^2)+\Delta_\epsilon\big]
		\nonumber\\
		&\quad
		-2\epsilon\Delta_\epsilon
		\big[-3\ell^8-3\ell^6r_h^2+3L^2r_h^6+\ell^4(-19L^2r_h^2+3r_h^4)
		+\ell^2(-4L^2r_h^4+3r_h^6)\big]
		\nonumber\\
		&\quad
		+2\epsilon L^4r_h^3
		\big[3\ell^{10}+3\ell^8L^2-30\ell^4L^2r_h^4+3L^2r_h^8
		+6\ell^6(4L^2r_h^2-r_h^4)
		+\ell^2(-16L^2r_h^6+3r_h^8)\big]
		\nonumber\\
		&\quad
		-\epsilon^3\ell^4
		\big[
		2L^4r_h^5(9\ell^4-6L^2r_h^2+9r_h^4-2\ell^2(7L^2-9r_h^2))
		+\Delta_\epsilon(-3\ell^4+r_h^2(L^2-3r_h^2)+3\ell^2(L^2-2r_h^2))
		\big]
		\nonumber\\
		&\quad
		-\epsilon^2\ell^2
		\big[
		-3\Delta_\epsilon(\ell^6+3r_h^4(L^2-r_h^2)-\ell^4(L^2+r_h^2)
		+\ell^2(6L^2r_h^2-5r_h^4))
		\nonumber\\
		&\qquad\qquad
		+L^4r_h^3(9\ell^8+29L^2r_h^6-3r_h^8-3\ell^6(L^2-20r_h^2)
		+\ell^4(-41L^2r_h^2+90r_h^4)
		+\ell^2(39L^2r_h^4+36r_h^6))
		\big].
	\end{align}
\end{widetext}
where
\begin{align}
	\Delta_{\epsilon}
	&=
	\sqrt{
		L^8 r_h^6
		\left[
		(\epsilon-1)^2\ell^4
		-2(1+3\epsilon)\ell^2 r_h^2
		+r_h^4
		\right]}.
\end{align}
For the branch used in the phase structure analysis below, the entropy obtained by integrating the first law is
\begin{widetext}
\begin{align}
	S_{\mathrm{II}}(r_h)=
	\frac{\pi}{G}
	\left[
	\frac{\Delta\left(6\ell^2 r_h+5r_h^3\right)}
	{2\left(\ell^2+r_h^2\right)}-
	\frac{3\left(-\ell^4+\ell^2 r_h^2+r_h^4\right)}
	{2\left(\ell^2+r_h^2\right)}
	+
	6\ell^2
	\log\left(
	\frac{3r_h+\Delta}{2(r_h-\Delta)}
	\right)
	\right],
	\qquad
	\Delta=\sqrt{r_h^2-8\ell^2},
	\qquad
	r_h>\sqrt{8}\ell .
	\label{eq:S-model-II}
\end{align}
\end{widetext}
The condition $r_h>\sqrt{8}\ell$ ensures that the entropy is real on the branch under consideration.

The free energy, defined by $F=M-TS$, takes the form
\begin{widetext}
	\begin{align}
		F_{\mathrm{II}}(r_h)
		=
		\frac{1}
		{4\ell^2 L^2 r_h^2(\ell^2+r_h^2)}
		\left[
		r_h^5
		\left(
		2\ell^4
		-
		2\ell^2(L^2-r_h^2)
		+
		L^2r_h(r_h-\Delta)
		\right)
		+
		\frac{4\mathcal A_{\mathrm{II}}\mathcal B_{\mathrm{II}}}{(\Delta-3r_h)^4}
		\right],
		\qquad
		\Delta=\sqrt{r_h^2-8\ell^2},
		\label{eq:F-model-II}
	\end{align}
	\begin{align}
		\mathcal A_{\mathrm{II}}
		=&
		12L^2r_h^6(r_h-\Delta)
		+
		12\ell^8(-6r_h+\Delta)
		-
		\ell^4
		\left[
		L^2r_h^2(103r_h-69\Delta)
		+
		18r_h^4(3r_h+\Delta)
		\right]\nonumber \\&+
		\ell^6
		\left[
		9r_h^2(-15r_h+\Delta)
		-
		4L^2(-34r_h+3\Delta)
		\right]
		+
		\ell^2
		\left[
		3r_h^6(3r_h-5\Delta)
		+
		5L^2r_h^4(-13r_h+3\Delta)
		\right],
		\label{eq:A-model-II}
	\end{align}
	\begin{align}
		\mathcal B_{\mathrm{II}}
		=&
		3\ell^4
		-
		3\ell^2r_h^2
		-
		3r_h^4
		+
		6\ell^2r_h\Delta
		+
		5r_h^3\Delta
		+
		12\ell^2(\ell^2+r_h^2)
		\log\left(
		\frac{3r_h+\Delta}{2(r_h-\Delta)}
		\right).
		\label{eq:B-model-II}
	\end{align}
\end{widetext}
As in Model I, the additive constant in the entropy is chosen consistently with the background normalization used for the free energy. The expressions above are the ones used to generate the representative phase diagrams for the AdS core branch. The corresponding thermodynamic volume for Model II is again
\begin{align}
	V_{\mathrm{II}}(r_h)=\frac{4\pi r_h^3}{3}.
\end{align}
The potential conjugate to the minimal-length scale is
\begin{widetext}
	\begin{align}
		\Psi_{\mathrm{II}}(r_h)
		&=
		\frac{1}{2\ell^3L^2r_h^2(\ell^2+r_h^2)^2}
		\Bigg[
		2\ell^2r_h^5(\ell^2+r_h^2)
		\left(
		2\ell^2+r_h^2
		+L^2\left[-1+\frac{2r_h}{\Delta}\right]
		\right)
		\nonumber\\
		&\quad
		-\ell^2r_h^5
		\left(
		2\ell^4
		-2\ell^2(L^2-r_h^2)
		+L^2r_h(r_h-\Delta)
		\right)
		\nonumber\\
		&\quad
		+r_h^5(\ell^2+r_h^2)
		\left(
		-2\ell^4
		+2\ell^2(L^2-r_h^2)
		+L^2r_h(-r_h+\Delta)
		\right)
		\nonumber\\
		&\quad
		+\frac{
			8\ell^2(\ell^2+r_h^2)
			\mathcal A_{\mathrm{II}}(r_h)
			\mathcal C_{\mathrm{II}}(r_h)
		}{
			\Delta(-3r_h+\Delta)^4(-r_h+\Delta)
			\left[-4\ell^2+r_h(-r_h+\Delta)\right]
		}
		\Bigg],
		\label{eq:Psi-model-II}
	\end{align}
	\begin{align}
		\mathcal C_{\mathrm{II}}(r_h)
		&=
		48\ell^6
		+66\ell^4r_h^2
		-49\ell^2r_h^4
		-19r_h^6
		+90\ell^4r_h\Delta
		+125\ell^2r_h^3\Delta
		+19r_h^5\Delta
		\nonumber\\
		&\quad
		-24(\ell^2+r_h^2)
		\left(
		16\ell^4
		-2\ell^2r_h(-3r_h+\Delta)
		+r_h^3(-r_h+\Delta)
		\right)
		\log(r_h-\Delta)
		\nonumber\\
		&\quad
		+12(\ell^2+r_h^2)
		\left(
		16\ell^4
		-2\ell^2r_h(-3r_h+\Delta)
		+r_h^3(-r_h+\Delta)
		\right)
		\log\left(4\ell^2+r_h^2-r_h\Delta\right).
	\end{align}
\end{widetext}
For the remaining models, the thermodynamic conjugate potentials can be obtained in the same way from Eq.~\eqref{eq:conjugates}. However, the resulting expressions do not provide additional conceptual insight. We therefore do not display them explicitly. Instead, we verify the Smarr relation directly in the numerical implementation used to generate the thermodynamic plots.

For Model III, the mass function is given by
\begin{align}
	m_{\mathrm{III}}(r_h)
	=
	\frac{
		r_h^3\left(r_h^4+L^2r_h^2-\ell^4\right)
	}
	{
		2L^2\left(r_h^4-\ell^4\right)
	}.
	\label{eq:m-model-III}
\end{align}
The Hawking temperature takes the form
\begin{align}
	T_{\mathrm{III}}(r_h)
	=
	\frac{
		3\ell^8
		-
		5\ell^4 L^2 r_h^2
		-
		6\ell^4 r_h^4
		+
		L^2 r_h^6
		+
		3r_h^8
	}
	{
		4\pi L^2 r_h^3\left(\ell^4+r_h^4\right)
	}.
	\label{eq:T-model-III}
\end{align}
Integrating the first law gives the entropy
\begin{align}
	S_{\mathrm{III}}(r_h)
	&=
	\pi
	\left[
	r_h^2
	\left(
	1+\frac{\ell^4}{\ell^4-r_h^4}
	\right)
	-
	\ell^2
	\log\left(
	\frac{r_h^2+\ell^2}{r_h^2-\ell^2}
	\right)
	\right].
	\label{eq:S-model-III}
\end{align}
The condition $r_h>\ell$ ensures that the logarithm is real on the black hole branch considered here. The corresponding free energy, $F=M-TS$, is
\begin{widetext}
	\begin{align}
		F_{\mathrm{III}}(r_h)
		=
		\frac{
			-6\ell^8 r_h^2
			+
			r_h^8\left(L^2-r_h^2\right)
			+
			\ell^4\left(10L^2r_h^4+11r_h^6\right)
			+
			\ell^2
			\left[
			3\ell^8
			+
			r_h^6\left(L^2+3r_h^2\right)
			-
			\ell^4\left(5L^2r_h^2+6r_h^4\right)
			\right]
			\log\left(
			\frac{r_h^2+\ell^2}{r_h^2-\ell^2}
			\right)
		}
		{
			4L^2r_h^3\left(\ell^4+r_h^4\right)
		}.
		\label{eq:F-model-III}
	\end{align}
\end{widetext}
As in the previous cases, the additive constant in the entropy is fixed consistently with the background normalization used in the free energy.

\begin{figure}[t]
	\centering
	\includegraphics[width=\columnwidth]{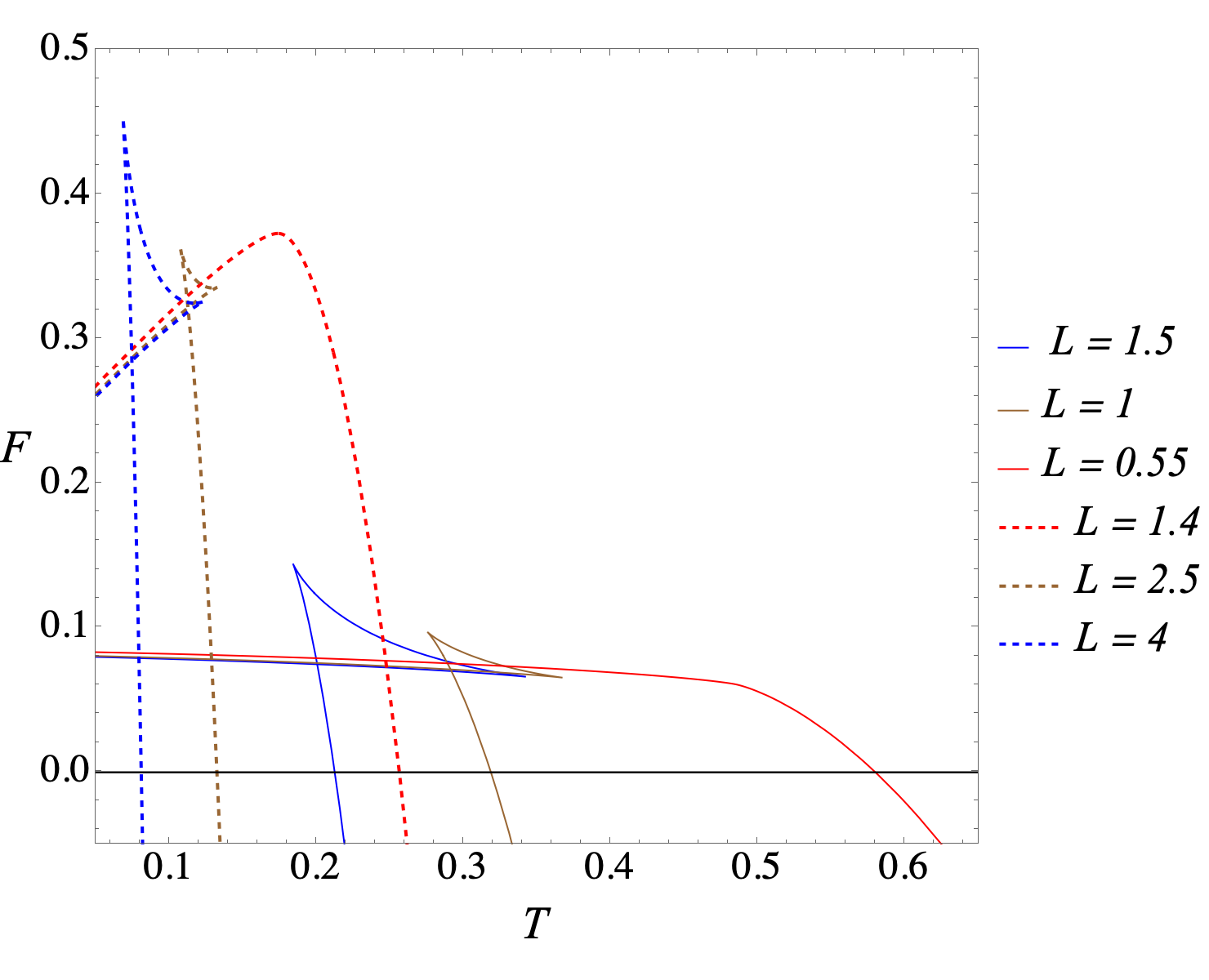}
	\caption{\textbf{Free energy: Models V and VI.} The free energy \(F\) is plotted parametrically as a function of the temperature \(T\), with the horizon radius used as the parameter. Solid curves correspond to Model V, which has a dS core, while dashed curves correspond to Model VI, which has an AdS core. The minimal-length scale is fixed to \(\ell=0.1\), and the cosmological-constant parameter \(L\) is varied. We work in units with \(G=1\).}
	\label{fig:PT-56}
\end{figure}

\begin{figure*}[t]
	\centering
	\includegraphics[width=\textwidth]{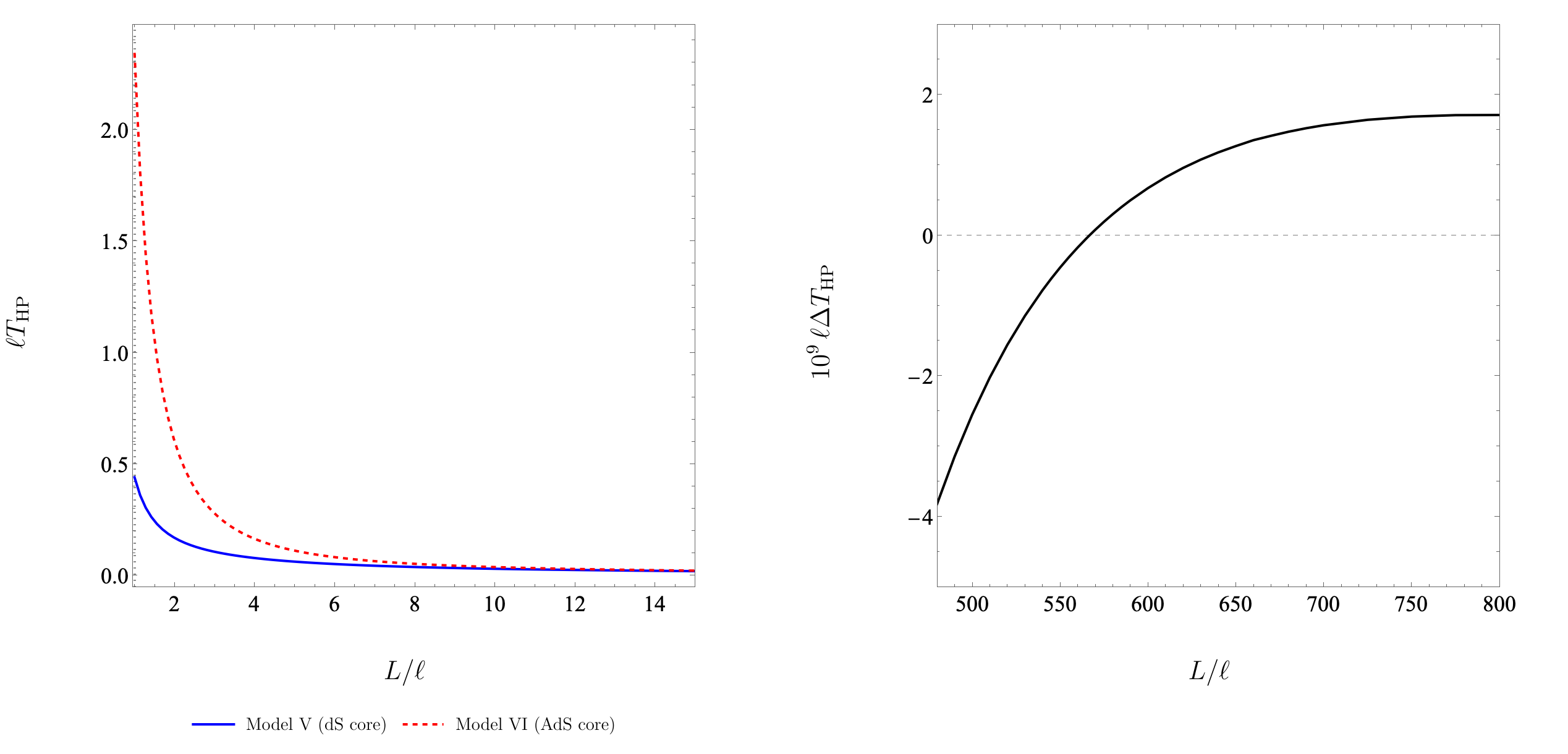}
	\caption{\textbf{Hawking--Page temperature for Models V and VI.}
		Left panel: dimensionless Hawking--Page temperature $\ell T_{\rm HP}$ as a function of the dimensionless AdS scale $L/\ell$ for Model V, which has a dS core, and Model VI, which has an AdS core. Right panel: corresponding dimensionless difference
		$\ell \Delta T_{\rm HP}=\ell\left(T^{\rm V}_{\rm HP}-T^{\rm VI}_{\rm HP}\right)$.
		The vertical dotted line marks the lower admissible boundary $L/\ell=1$, required for the correct asymptotic AdS behaviour as shown in Table~\ref{tab:polymerization-models}. The ordering of the transition temperatures reverses close to this boundary, while for sufficiently large $L/\ell$ one finds $T^{\rm V}_{\rm HP}>T^{\rm VI}_{\rm HP}$.}
	\label{fig:Thp-V-VI}
\end{figure*}
For Model V, the mass, temperature, and entropy are given by
\begin{align}
	m_{\mathrm{V}}(r_h)
	=
	\frac{
		-\ell^4 r_h^3
		+r_h^7
		+L^2 r_h^3\sqrt{r_h^4-\ell^4}
	}{
		2L^2\left(r_h^4-\ell^4\right)
	},
	\label{eq:m-model-V}
\end{align}
\begin{widetext}
	\begin{align}
		T_{\mathrm{V}}(r_h,\ell,L)
		&=
		\frac{
			3\ell^8
			+3r_h^8
			+L^2 r_h^4\sqrt{r_h^4-\ell^4}
			-3\ell^4\left(2r_h^4+L^2\sqrt{r_h^4-\ell^4}\right)
		}{
			4\pi L^2 r_h^5\sqrt{r_h^4-\ell^4}
		}.
		\label{eq:T-model-V}
	\end{align}
\end{widetext}
and
\begin{align}
	S_{\mathrm{V}}(r_h,\ell,L)
	=
	\frac{
		\pi\left(r_h^4-2\ell^4\right)
	}{
		\sqrt{r_h^4-\ell^4}
	},
	\qquad
	r_h>\ell .
	\label{eq:S-model-V}
\end{align}
The condition $r_h>\ell$ selects the real black hole branch. The corresponding free energy $F=M-TS$ is
\begin{widetext}
\begin{align}
	F_{\mathrm{V}}(r_h)
	=
	\frac{
		-6\ell^8
		-r_h^8
		+L^2 r_h^4\sqrt{r_h^4-\ell^4}
		+\ell^4\left(9r_h^4+6L^2\sqrt{r_h^4-\ell^4}\right)
	}{
		4L^2 r_h^5
	}.
	\label{eq:F-model-V}
\end{align}
\end{widetext}
For the remaining models, solving the horizon condition explicitly for the mass as a function of the horizon radius is either not possible in closed analytic form or leads to expressions that are too cumbersome to be useful. Even in cases where an analytic solution can in principle be obtained, for example through Cardano’s formula, the resulting expressions do not provide additional physical insight. We therefore treat these models numerically whenever needed.

\section{Phase structure of various models} \label{sec:app:PT}

In the main text, we focused primarily on Models I and II. This choice was motivated by the fact that Model I corresponds to the well-known Hayward-type RBH, while Model II provides its AdS core counterpart. Our main conclusions were therefore drawn from this representative pair. However, in order to test the robustness of these conclusions and to provide further evidence for their possible universality, it is useful to examine additional examples within the same class of polymerized vacuum models.

This is the purpose of the present appendix section. We extend the analysis to Models III, IV, V, and VI. Models VII and VIII are not considered here because their thermodynamic expressions are significantly more complicated and do not add qualitatively new features to the discussion. As shown in Appendix~\ref{sec:app:thermo}, only some of the models admit analytic expressions for the basic thermodynamic quantities. For this reason, the plots presented in this appendix section are obtained by numerical evaluation. This is mainly due to the fact that, for these models, the mass cannot in general be written as a simple analytic function of the horizon radius.

The same qualitative picture persists for the additional model pairs considered in this appendix. For Models III and IV, the free energy as a function of temperature is shown in Fig.\ref{fig:PT-34}. As in the representative case discussed in the main text, the thermodynamically relevant transition is the Hawking–Page transition, determined by the crossing of the black hole free energy with the thermal AdS background. The small/large black hole branch structure is still present, but it occurs at positive free energy and therefore represents only a subdominant saddle structure rather than the physical first-order transition of the canonical ensemble. The corresponding Hawking–Page temperatures, obtained numerically and displayed in Fig.\ref{fig:Thp-III-IV}, again show that the dS core solution has a higher transition temperature than its AdS core counterpart.

This conclusion is further supported by Models V and VI. Their free-energy diagrams and Hawking–Page temperatures are shown in Figs.~\ref{fig:PT-56} and \ref{fig:Thp-V-VI}, respectively. No new qualitative phase behavior appears: the Hawking–Page transition remains the dominant transition, while the dS core geometry again exhibits a higher Hawking–Page temperature than the corresponding AdS core geometry. These additional examples therefore indicate that the main thermodynamic conclusions of Sec.\ref{sec:app:thermo} are not special to Models I and II, but persist across a broader class of polymerized vacuum RBHs.

\end{document}